\documentclass[aps,prb,superscriptaddress,twocolumn]{revtex4-2}
\usepackage{amsfonts, amssymb, amsmath, amsthm, dsfont, graphicx}
\usepackage{mathrsfs}
\usepackage{mathtools}
\usepackage{subfigure}
\usepackage{multirow}
\usepackage{dcolumn}
\usepackage{bm}
\usepackage{color}
\usepackage{braket}
\usepackage[usenames,dvipsnames]{xcolor}
\usepackage[all]{xy}
\usepackage[normalem]{ulem}
\usepackage{comment}
\usepackage{environ}

\usepackage{array}

\newcommand{\hc}{{\rm h.c.}}
\newcommand{\meff}{m_{\rm eff}}
\newcommand{\tR}{{\widetilde R}}
\newcommand{\mR}{{\mathcal R}}
\newcommand{\tL}{{\widetilde \Lambda}}
\newcommand{\br}{{\boldsymbol{\rho}}}

\renewcommand{\Re}{{\rm Re}}
\renewcommand{\Im}{{\rm Im}}

\renewcommand{\bf}[1]{\textnormal{\textbf{#1}}}

\newcommand{\tr}{\mathrm{Tr}}

\newtheorem{proposition}{Proposition}

\definecolor{Nathanblue}{rgb}{0.,0.24,0.51}

\newcommand{\blue}{\color{Nathanblue}}

\usepackage[colorlinks=true,allcolors=blue]{hyperref}

\newif\ifSM

\SMtrue

\ifSM
\NewEnviron{supplementalMaterials}{%

	\setcounter{figure}{0}
	\setcounter{equation}{0}
	\setcounter{table}{0}
	\setcounter{page}{1}
	\BODY
}
\else
\excludecomment{supplementalMaterials}
\fi

\makeatletter
\def\maketitle{
	\@author@finish
	\title@column\titleblock@produce
	\suppressfloats[t]}
\makeatother

\begin{document}
	
	\title{{\blue Hall viscosity from metric-sensitive dichroic probes}}
	
	\author{Alberto Nardin} 
\thanks{These authors contributed equally to this work.}
\affiliation{Universit\'{e} Paris-Saclay, CNRS, LPTMS, 91405, Orsay, France}

\author{Bruno Mera}
\thanks{These authors contributed equally to this work.}
\affiliation{Instituto de Telecomunica\c{c}\~oes and Departmento de Matem\'{a}tica, Instituto Superior T\'ecnico, Universidade de Lisboa, Avenida Rovisco Pais 1, 1049-001 Lisboa, Portugal}
\affiliation{Advanced Institute for Materials Research (WPI-AIMR), Tohoku University, Sendai 980-8577, Japan}

\author{Anaïs Defossez}
\affiliation{Laboratoire Kastler Brossel, Collège de France, CNRS, ENS-Université PSL,
Sorbonne Université, 11 Place Marcelin Berthelot, 75005 Paris, France}
\affiliation{International Solvay Institutes, B-1050 Brussels, Belgium}
\affiliation{Center for Nonlinear Phenomena and Complex Systems,
Université Libre de Bruxelles, CP 231, Campus Plaine, B-1050 Brussels, Belgium}

\author{Baptiste~Bermond}
\affiliation{Laboratoire Kastler Brossel, Collège de France, CNRS, ENS-Université PSL,
Sorbonne Université, 11 Place Marcelin Berthelot, 75005 Paris, France}

\author{Tomoki Ozawa}
\affiliation{Advanced Institute for Materials Research (WPI-AIMR), Tohoku University, Sendai 980-8577, Japan}
\affiliation{RIKEN Center for Interdisciplinary Theoretical and Mathematical Sciences (iTHEMS), RIKEN, Wako, Saitama 351-0198, Japan}

\author{Nathan Goldman}
\email{nathan.goldman@lkb.ens.fr}
\affiliation{Laboratoire Kastler Brossel, Collège de France, CNRS, ENS-Université PSL,
Sorbonne Université, 11 Place Marcelin Berthelot, 75005 Paris, France}
\affiliation{International Solvay Institutes, B-1050 Brussels, Belgium}
\affiliation{Center for Nonlinear Phenomena and Complex Systems,
Université Libre de Bruxelles, CP 231, Campus Plaine, B-1050 Brussels, Belgium}
%

\begin{abstract}

Hall viscosity characterizes the geometric response of a quantum Hall droplet to deformations of the underlying metric, yet it has remained difficult to measure directly. We propose a spectroscopic probe based on circular dichroism, using chiral metric-sensitive drives -- implemented as rotating quadrupolar (``saddle") perturbations -- that effectively modulate the metric and couple to the generators of area-preserving deformations. The resulting dichroic signal directly measures the Hall viscosity, while frequency-resolved spectroscopy disentangles it from other excitations. A local formulation further enables spatially resolved markers of Hall viscosity applicable to both continuum and lattice systems. Our results open a direct route to measuring Hall viscosity in quantum-engineered platforms such as cold atoms in optical lattices.
\end{abstract}

\maketitle

{\it Introduction---} The topological characterization of quantum Hall states has traditionally rested on the quantized Hall response and the related Chern number~\cite{PhysRevLett.45.494,PhysRevLett.49.405,PhysRevB.31.3372,PhysRevB.23.5632}, which successfully classifies integer and fractional quantum Hall states according to their global electromagnetic response. Yet distinct quantum Hall states can share the same Chern number while differing in their internal correlations~\cite{wen_zee,wen1995topological}, signaling the need for additional characterization. A key refinement is provided by the Hall viscosity~\cite{avron:seiler:zograf:1995,Avron1998}, which is directly related to the intrinsic orbital spin $s$ of the many-body state via $\eta_H = \hbar n_0 s/2$~\cite{wen2012modular,read2009non,read:rezayi:2011,fremling2014hall,cho2014geometry,arouca2022quantum,cappelli2018bulk}. This quantity distinguishes phases sharing the same Hall conductance — including the Moore--Read Pfaffian~\cite{moore1991nonabelions} and anti-Pfaffian states~\cite{levin2007particle,lee2007particle,Zhu_interface} — and carries direct geometric meaning: arising as the linear response to adiabatic metric deformations, it encodes Berry curvature in the space of geometric deformations~\cite{avron:seiler:zograf:1995,kimura2021hall} and provides access to the ``shape" of many-body correlations~\cite{Zaletel_2013,Tu_2013}, complementing topological invariants with geometric information.

Despite its fundamental importance, experimental access to Hall viscosity remains challenging. Signatures consistent with odd-viscous response have been reported in the hydrodynamic electron fluid of graphene~\cite{berdyugin2019measuring} and in chiral active fluids~\cite{banerjee2017odd,soni2018free}. Existing strategies rely on transport-based signatures~\cite{delacretaz2017transport,pellegrino2017nonlocal,HoyosSon} or hydrodynamic effects~\cite{scaffidi2017hydrodynamic,alekseev2016negative,holder2019unified,rao2020hall,Musser_2024}, while more recent dynamical approaches include conformal drives in $(1+1)\text{D}$ systems~\cite{lapierre2025nonequilibrium} and geometric quenches probing chiral graviton modes~\cite{haldane2009hall,Haldane_2011,golkar2016spectral,son2019chiral,liang2024evidence,xavier2025chiral,bacciconi2026chiralgravitonmodesnonabelian}.

In this work, we introduce a probe of Hall viscosity based on circular dichroism (CD), an approach well suited to extract quantum geometric properties such as Berry curvatures, quantum metrics~\cite{tran2017probing,ozawa:goldman:2018,Repellin_dichroism,unal_2025,bermond2025local} and magnetizations~\cite{bermond2025dichroism}. We show that suitably engineered circular drives — realized as rotating quadrupolar perturbations — effectively modulate the metric and couple to the generators of area-preserving deformations (APD). The polarization of these drives selects the chirality of the excitation, so that the differential absorption between the two drives directly measures the chiral response to metric deformations, i.e., the Hall viscosity. Frequency-resolved measurements further allow one to isolate the Hall-viscosity contribution from other responses, and the approach naturally admits a local formulation leading to spatially resolved markers applicable to lattice systems~\cite{Barkeshli_2012,Shapourian_2015,Tuegel_2015,rao2020hall}. Our scheme builds on Ref.~\cite{sisterEllipticpaper}, where dichroic responses were shown to encode the geometry of anisotropic quantum Hall droplets.

{\it Theoretical background---} 
We start by introducing the Hall viscosity framework. We consider a system of non-interacting fermions of mass $\meff$, moving on a plane under the effect of an orthogonal magnetic field, $B\hat z=\boldsymbol{\nabla}\times{\bf A}$;
the single-particle problem is notoriously described by the Landau Hamiltonian
\begin{equation}
	\label{eq:ll_limit}
	H = \frac{1}{2\meff} g^{ab} \pi_a  \pi_b,
\end{equation}
where $\pi_a=p_a-e A_a$ are the kinetic momenta, which obey $[\pi_x,\pi_y]=i \hbar e B$.
In general, a diagonal unimodular Landau-orbit (inverse) metric $g^{ab}$ different from the Euclidean one $\delta^{ab}$ can appear in this kinetic term~\cite{bo2012prb}; with a simple rotation, such a metric can be written as a diagonal matrix $(g^{ab})={\rm diag}(\beta,\beta^{-1
})$. The spectrum of Eq.~\eqref{eq:ll_limit} is easily found, yielding flat Landau levels $E_{n}=\hbar\frac{eB}{\meff}(n+\frac{1}{2})$, independent of the underlying geometry set by the metric.

We follow Park and Haldane~\cite{Park_2014} and define $\tR^a=\epsilon^{ab}\pi_b/eB$, in terms of which one introduces APD generators~\cite{Park_2014}
$\tL^{ab}= \frac{1}{4\ell_B^2}\{\tR^a,\tR^b\}$, 
which satisfy the commutation relations of a $\mathfrak{sl}(2;\mathbb R)$ algebra $[\tL^{ab},\tL^{cd}]=\frac{i}{2}(\epsilon^{ac}\tL^{bd}+\epsilon^{ad}\tL^{bc}+ a\leftrightarrow b)$.
The Landau-orbit metric can be deformed -- without changing its determinant -- by the action of the APD generators $U_\alpha=e^{i \alpha_{ab}\tL^{ab}}$. In particular, the Hamiltonian corresponding to any flat metric can be expressed in terms of the original one as $H_\alpha=U^\dagger_\alpha H U_\alpha$, parametrized by the variables $\alpha_{ab}$. In this framework, Hall viscosity is identified as the adiabatic response of the incompressible ground state $\ket{\Psi_0}$ to this deformation, obtained by promoting the parameters to slowly varying functions of time, $\alpha_{ab}\to\alpha_{ab}(t)$.
If we define the deformed state  $\ket{\Psi(\alpha)}=U_\alpha\ket{\Psi_0}$, the generalized force conjugate to the small deformation $\alpha_{ab}$ is given by the adiabatic theorem~\cite{avron:seiler:zograf:1995, Park_2014, De_Grandi_2010, Weinberg_2017}
\begin{equation}
	\label{eq:stress_tensor}
	F^{ab} =  -\Braket{\frac{\partial H_\alpha}{\partial \alpha_{ab}}} = -\left.\frac{\partial E}{\partial \alpha_{ab}}\right|_{\alpha=0} + {\mathcal F}^{abcd}\,\dot{\alpha}_{cd} \, ,
\end{equation}
where $\mathcal F$ is a Berry curvature in the deformation space 
\begin{equation}
	\begin{split}
		{\mathcal F}^{abcd} &= -2\hbar\, \Im \left.\Braket{\frac{\partial \Psi}{\partial\alpha_{ab}}|\frac{\partial \Psi}{\partial\alpha_{cd}}}\right|_{\alpha=0} \\&= i\hbar \braket{\Psi_0|[\tL^{ab},\tL^{cd}]|\Psi_0}.
	\end{split}
\end{equation}
One defines the Landau-orbit Hall-viscosity tensor $\widetilde\eta^{abcd} = -\frac{1}{A}{\mathcal F}^{abcd}$ by dividing by the area $A$ of the quantum Hall system, 
which in terms of the number $N_\Phi$ of flux quanta penetrating it reads
$A= N_\Phi 2\pi\ell_B^2$.
Using the APD commutation relations, the 4-tensor can be broken up into a symmetric 2-tensor $\widetilde\eta_{H}^{ab}$
\begin{align}
	\label{eq:viscosity_tensor}
	\widetilde\eta^{abcd} &= \frac{1}{2}(\epsilon^{ac}\widetilde\eta_H^{bd}+\epsilon^{ad}\widetilde\eta_H^{bc}+ a\leftrightarrow b)
	\\
	\widetilde\eta_H^{ab} &= \frac{\hbar}{2\pi \ell_B^2}\,\frac{1}{N_\Phi}\,\braket{\Psi_0|\tL^{ab}|\Psi_0}.
\end{align}

For the Hamiltonian in Eq.~\eqref{eq:ll_limit}, when the $N$-particle state $\ket{\Psi_0}$ occupies the lowest $\mathcal C$ Landau levels, this evaluates to
\begin{equation}
	\widetilde\eta_H^{ab} = \frac{\hbar}{8\pi \ell_B^2}\,\mathcal C^2g^{ab} \, ,
\end{equation}
which is the Hall viscosity first discussed by Avron, Seiler and Zograf~\cite{avron:seiler:zograf:1995,Avron1998}. In this work, we focus on the scalar
\begin{equation}
	\label{eq:contracted_viscosity}
	\overline\eta_H = \frac{1}{2}g_{ab}\widetilde\eta_H^{ab} = \frac{\hbar}{8\pi \ell_B^2} \mathcal C^2\, ,
\end{equation}
obtained by contracting the Hall-viscosity 2-tensor with the metric.

While the above discussion focused on Landau levels in the continuum, we will later explore the Hall viscosity and its dichroic measurements using a lattice description, considering the Harper-Hofstadter model relevant to a broad range of quantum-simulation platforms~\cite{aidelsburger2015measuring,leonard2023realization,ozawa:price:amo:goldman:etal:2019}. This choice is motivated not only by experimental practicability, but also by the opportunity to study lattice effects beyond the ideal Landau levels of Eq.~\eqref{eq:ll_limit}.

{\it The dichroic approach: viscosity through shaking---} 
We now show how the Landau-orbit Hall-viscosity  $\overline\eta_H $ can be accessed through a spectroscopic probe, i.e.~by “shaking” the ground-state $\ket{\Psi_0}$.

First, we use the fact that the Hall viscosity can be written as the expectation value of the Landau-orbit generator of rotations~\cite{Haldane_2011}, defined as $L\equiv g_{ab}\tL^{ab}=\frac{\ell_B^2}{4\hbar^2}g^{ab}\{\pi_a,\pi_b\}$: $\overline\eta_H = \frac{\hbar}{4\pi \ell_B^2 N_\Phi} \langle \Psi_0| L |\Psi_0\rangle$.
Then, we decompose the unimodular metric in terms of two vectors~\cite{Haldane_2011,Qiu_2012}, $g^{ab}=e_1^a e_1^b+e_2^ae_2^b$, with $e_1^a e_2^b \epsilon_{ab}=1$, from which we define “flat-space” momenta, $\Pi_i=e_i^a \pi_a$. 
With this decomposition, the rotation generator simplifies as
$L=\ell_B^2\frac{\Pi_1^2+\Pi_2^2}{2\hbar^2}$. 
	
In view of defining a proper metric-sensitive dichroic probe, we introduce two key operators
\begin{align}
		G_0 = \frac{\ell_B^2}{\hbar^2}(\Pi_1^2 - \Pi_2^2), \qquad G_1 = \frac{\ell_B^2}{\hbar^2}(\Pi_1\Pi_2+\Pi_2\Pi_1) \, ,\label{eq:G0G1}        
\end{align}
whose commutator produces a rotation, $[G_0,G_1]= 8 i L$. Crucially, the expectation value of this commutator then yields the Hall viscosity $\overline\eta_H$, forming the basis of our spectroscopic scheme.
\begin{figure}[t]
	\includegraphics[width=1\linewidth]{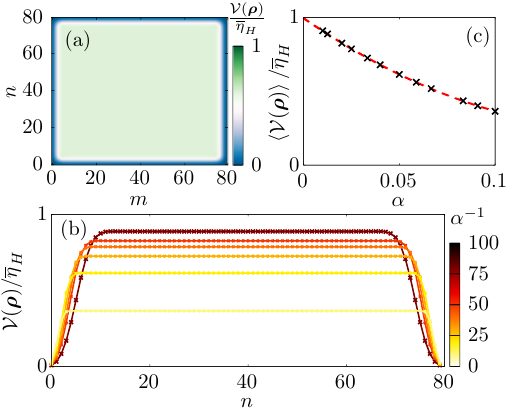}
	\caption{
		The Landau-orbit Hall viscosity, as extracted from the viscosity marker $\mathcal V(\br)$ defined in Eq.~\eqref{eq:viscosity_marker}, for the non-interacting fermionic Harper-Hofstadter model with $N_x=N_y=80$ lattice sites and $J_x=J_y=1$.
		(a) The viscosity marker for $N=280$ fermions filling up the lowest band at flux  $\alpha=1/20$;
		a slice of the marker at constant $m=N_x/2$ is compared to the theoretical value [Eq.~\eqref{eq:uniform_marker_value}] in panel (b), for various values of flux $\alpha$.
		(c) Value of the marker averaged over a $10\times10$ square located around $\br\!=\!(N_x/2,N_y/2)$ (black crosses), as a function of flux $\alpha$, showing how it extrapolates to the Landau-orbit viscosity as $\alpha\rightarrow0$. The points have been fitted with a quadratic function of $\alpha$ (red dashes); as $\alpha\rightarrow0$, we find $\braket{\mathcal V(\br)}/\overline\eta_H=0.995(1)$.
	}
	\label{fig:viscosity_marker}
\end{figure} 
Indeed, inspired by circular dichroic probes used to measure the Hall conductivity~\cite{tran:dauphi:grushin:zoller:goldman:2017,asteria2019measuring,repellin2019detecting,unal_2025,sisterEllipticpaper}, we now introduce the metric-sensitive dichroic probe
\begin{equation}
	\label{eq:bare_dichroic_probe}
	U_\pm= {\rm Re}\left[ u_0 e^{-i \omega t} (G_0\mp i G_1)\right],
\end{equation}
with the goal of spectroscopically accessing the Hall viscosity. Here, the operators $G_0\pm i G_1$ play a role analogous to the circularly-polarized electric fields in conventional circular dichroism~\cite{tran:dauphi:grushin:zoller:goldman:2017}.
The excitation rates resulting from the perturbation in Eq.~\eqref{eq:bare_dichroic_probe} can be computed using Fermi's golden rule,
\begin{equation}
	\label{eq:fgr_rates}
	\Gamma_\pm(\omega>0) = 2\pi \frac{|u_0|^2}{4\hbar^2}\sum_{j\neq0} |\braket{\Psi_j|G_0\mp i G_1|\Psi_0}|^2\delta(\omega_{j,0}-\omega).
\end{equation}
where $\ket{\Psi_j}$ denote all possible excited states and $\omega_{j,0}$ the corresponding energies. One finally obtains that the integrated differential rate
gives direct access to the Hall viscosity $\overline\eta_H$,
\begin{equation}
	\label{eq:quantized_rate_difference}
	\frac{\Delta\Gamma^{\rm int}}{A\,|u_0|^2} \equiv \frac{\int_0^\infty (\Gamma_+(\omega)-\Gamma_-(\omega))d\omega}{A\,|u_0|^2} = \frac{16\pi}{\hbar^3} \, \overline\eta_H.
\end{equation}
This is one of the key results of our approach. Before turning to its practical implications, we note that other polarizations are  conceivable for the probe in Eq.~\eqref{eq:bare_dichroic_probe}. In particular, a linear drive of the form 
$U\sim G_{0,1} \cos(\omega t)$
would probe the quantum metric defined in metric-deformation space~\cite{ozawa:goldman:2018,haldane2009hall}. This observation points toward a broader generalization of the SWM sum rule~\cite{souza2000prb,bermond2025dichroism} to viscosity-type transport coefficients.

{\it Implementation---}
We now turn to the explicit implementation of the metric-sensitive dichroic probe.
With the coordinate choice that makes $g^{ab}$ diagonal, one can choose $e_1=(\beta^{1/2},0)$ and $e_2=\left(0,\beta^{-1/2}\right)$, which yields
$G_0=\beta \pi_x^2 - \beta^{-1}\pi_y^2$ and $G_1=\pi_x\pi_y+\pi_y\pi_x$. The perturbation in Eq.~\eqref{eq:bare_dichroic_probe} can then be induced by modifying the metric of Eq.~\eqref{eq:ll_limit} in a time-dependent way; for instance, assuming $u_0$ to be real,
\begin{equation}
	\label{eq:metric_deformation_viscosity_probe}
	g^{-1}_\mp(t) = \begin{pmatrix}
		\beta(1+\delta \cos(\omega t)) & \mp \delta \sin(\omega t)\\
		\mp\delta \sin(\omega t) & \beta^{-1}\left(1-\delta \cos(\omega t)\right),
	\end{pmatrix}
\end{equation}
where $\delta = 2\meff u_0 \ell_B^2/\hbar^2$. Notice that $\det g_\mp = 1-\delta^2$, i.e.~the metric is only unimodular at linear order in $\delta$.

As already motivated, we now introduce the Harper-Hofstadter model~\cite{Harper_1955, Hofstadter_1976} as a concrete setup for exploring lattice implementations of our scheme. The corresponding Hamiltonian reads
\begin{align}
	\label{eq:hh_hamiltonian}
	H_{\rm HH}=&-J_x \sum_{m,n} c_{m+1,n}^\dagger c_{m,n}^{} -J_y \sum_{m,n} e^{i 2\pi \alpha m}c_{m,n+1}^\dagger c_{m,n}^{}\nonumber \\
	&+\hc,
\end{align}	
where $c_{m,n}^{(\dagger)}$ are the annihilation (creation) operators of a single fermion at the  lattice site $(m,n)$ of a square lattice, and $2\pi \alpha$ the magnetic flux per plaquette.
In the continuum limit where the magnetic length exceeds the lattice spacing, $\ell_B=a_0/\sqrt{2\pi\alpha} \gg a_0$, the model reduces to the Landau Hamiltonian of Eq.~\eqref{eq:ll_limit} with effective mass  $\meff^{-1}=2a_0^2\sqrt{J_xJ_y}/\hbar^2$ and  diagonal unimodular Landau-orbit metric
$g^{xx}=1/g^{yy}=\beta$, where $\beta=\sqrt{J_x/J_y}$ encodes the hopping anisotropy~\cite{supp_mat}. We show in Ref.~\cite{supp_mat} how the metric deformation of Eq.~\eqref{eq:metric_deformation_viscosity_probe} can be explicitly implemented via time-dependent modulation of the hopping amplitudes in an optical triangular lattice~\cite{supp_mat}.

{\it A local marker for the Hall viscosity---}
We now introduce a local marker for the Landau-orbit Hall viscosity, by writing the excitation rates in Eq.~\eqref{eq:fgr_rates} as a sum of local contributions. This follows the approach of Ref.~\cite{tran:dauphi:grushin:zoller:goldman:2017}, which linked circular-dichroic signals to the Bianco-Resta Chern marker~\cite{BiancoResta_2011}. Although experimentally accessing this marker would require monitoring excitation rates (or  absorbed power) with single‑site resolution -- a challenging task -- its theoretical formulation nevertheless provides fundamental insight into Hall viscosity and related geometric responses, as we now explain.

We write the differential integrated rates as a sum of local contributions, $\Delta\Gamma^{\rm int}=|u_0|^2\left(16\pi \hbar^{-3}\sum_\br \mathcal V(\br)\,a_0^2 \right)$, where $\br$ is a lattice-site index. From Eq.~\eqref{eq:fgr_rates}, one obtains
\begin{equation}
	\label{eq:viscosity_marker}
	\mathcal V(\br)=\frac{\hbar}{8a_0^2} \,\Im\bra{\br}P G_0 Q G_1 P \ket{\br} \, ,
\end{equation}
where $\ket{\br}$ is a well-localized state at the site $\br$, while $Q$ ($P$) is a projector onto the empty (occupied) states~\cite{supp_mat}.
Figure~\ref{fig:viscosity_marker}(a) shows that $\mathcal V(\br)$ is uniform in the bulk, where its value reflects the Landau-orbit viscosity through
\begin{equation}
	\label{eq:uniform_marker_value}
	\mathcal V(\br) = \overline\eta_H = \frac{\hbar\alpha}{4 a_0^2} \mathcal C^2.
\end{equation}
Interestingly, if one assumes the marker to be uniform over the entire system's area, we obtain $\sum_\br\mathcal V(\br)\,a_0^2=\overline\eta_H A$, hence recovering the result in Eq.~\eqref{eq:quantized_rate_difference}.

Figure~\ref{fig:viscosity_marker}~(b,c) shows that this marker converges to the expected Hall viscosity value [Eq.~\eqref{eq:contracted_viscosity}] in the continuum limit $\alpha \rightarrow0$, yet exhibits deviations at finite $\alpha$. We attribute these to the lattice discretization of the kinetic momenta entering Eq.~\eqref{eq:viscosity_marker}. An analogous effect arises in the Bianco-Resta Chern marker, $\mathfrak C(\br)=4\pi \Im\bra{\br}PxQyP\ket{\br}$:~as we discuss in Ref.~\cite{supp_mat}, replacing the coordinates $r^a$ in its definition with kinetic momenta produces identical finite-$\alpha$ deviations, confirming that discretization of the momenta is the common source of error. This motivates reformulating the Hall-viscosity marker of Eq.~\eqref{eq:viscosity_marker} using coordinates $r^a$ alone. As we show below, this yields markedly more robust results at finite $\alpha$, and points towards a more accessible experimental probe of Hall viscosity.

\begin{figure}[t]
	\includegraphics[width=1\linewidth]{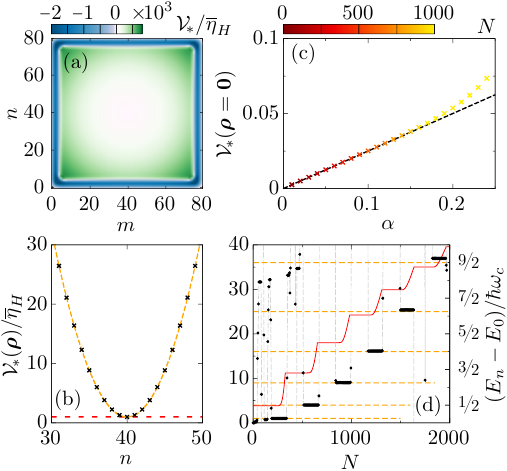}
	\caption{
		The Landau-orbit Hall viscosity, as extracted from the viscosity marker $\mathcal V_{\!*}(\br)$ defined in Eq.~\eqref{eq:viscous_marker_real_space}, for the Harper-Hofstadter model with $N_x=N_y=80$ lattice sites and $J_x/J_y=1$; unless specified, $\alpha\!=\!1/20$.
		(a) The viscosity marker $\mathcal V_{\!*}(\br)$ for $N=280$ fermions filling up the lowest band;
		a slice of the marker in (a), at constant $m=N_x/2$, is compared to the theoretical prediction [Eq.~\eqref{eq:viscous_marker_real_space}] (orange dashes) in (b); the constant offset (red dashes) matches the Hall viscosity [Eq.~\eqref{eq:contracted_viscosity}].
		(c) Value of the marker at $\br=(N_x/2,N_y/2)\equiv\bf0$ (crosses), compared to the Landau-orbit viscosity [Eq.~\eqref{eq:contracted_viscosity}] (black dashes) as a function of the flux $\alpha$.
		(d) Value of the marker at $\br=\bf0$ (black circles joined by gray dashes), compared to the Landau-orbit viscosity [Eq.~\eqref{eq:contracted_viscosity}] (orange dashes) for various values of $\mathcal C$ (number of filled bands), 
		as a function of the number of fermions $N$. The single-particle energies (red line) are shown on the right $y$-axis; these have been shifted by $E_0=-2(J_x+J_y)$~\cite{supp_mat}, and scaled by the cyclotron energy $\hbar\omega_c=eB/\meff$.
	}
	\label{fig:viscosity_real_space_marker}
\end{figure} 

{\it A different approach: the rotating saddle---}
The form of the probe in  Eq.~\eqref{eq:bare_dichroic_probe} leads us to consider a conceptually simpler rotating quadrupolar potential
\begin{equation}
	\label{eq:rotating_saddle}
	V_\pm= {\rm Re}\left[ v_0 e^{-i \omega t} ({\mathcal Q}_0\mp i {\mathcal Q}_1)\right],
\end{equation}
where $\ell_B^2{\mathcal Q}_0 = (X^1)^2 - (X^2)^2$, $\ell_B^2{\mathcal Q}_1 = 2X^1X^2$ have the same structure of Eq.~\eqref{eq:G0G1} but in real space:~$X^i=E^i_a r^a$ are defined in terms of particle's coordinates $r_a$, and $E_a^i = \epsilon_{ab} \epsilon^{ij} e^b_j$.
With the coordinate choice that makes $g^{ab}$ diagonal, $\ell_B^2{\mathcal Q}_0=\beta^{-1}x^2-\beta y^2$ and $\ell_B^2{\mathcal Q}_1=2xy$ -- the perturbation defined by Eq.~\eqref{eq:rotating_saddle} is a rotating saddle.

For a finite system with edges, the  integrated differential rate associated with the probe in Eq.~\eqref{eq:rotating_saddle} must vanish~\cite{stringari_pitaevskii}, yet spatial or frequency selection gives direct access to the Landau-orbit Hall viscosity; this is analogous to the detection of the Hall conductance via circular-dichroic rates under open boundaries~\cite{unal_2025}. This can be understood heuristically from the fact that the real-space perturbation in Eq.~\eqref{eq:rotating_saddle} ``contains'' the metric-shaking one [Eq.~\eqref{eq:bare_dichroic_probe}].
To see this, one introduces guiding center coordinates $R^i=r^i+\frac{\ell_B^2}{\hbar}\epsilon^{ij}\pi_j$, which commute with the kinetic momenta and satisfy $[R^x,R^y]=-i\ell_B^2$; 
the coordinates then split into guiding-center and kinetic momenta
$X^i = \mR^i - \frac{\ell_B^2}{\hbar}\epsilon^{ij}\Pi_j$, where ${\mR}^i=E^i_a R^a$. As a consequence,
${\mathcal Q}_i = - G_i $ (+ other operators), and thus $U_\pm$ in Eq.~\eqref{eq:bare_dichroic_probe} can be interpreted as a component of $V_\pm$ in Eq.~\eqref{eq:rotating_saddle}.
The key advantage of the real-space probe Eq.~\eqref{eq:rotating_saddle} over the metric-shaking protocol Eq.~\eqref{eq:bare_dichroic_probe} lies in the fact that $V_{\pm}$ does not directly involve kinetic momenta; as such, its definition on a finite lattice is straightforward.

Using both general semiclassical arguments -- based on energy conservation and Kramers--Kronig relations -- as well as an explicit Landau-level calculation, we show in Ref.~\cite{supp_mat} that the local marker $\mathcal V_{\!*}(\br)$, defined via the integrated differential rate $\Delta\Gamma^{\rm int}_*=|v_0|^2\left(\frac{16 \pi}{\hbar^3}\sum_{\br}\mathcal V_{\!*}(\br) a_0^2\right)$ induced by the rotating-saddle potential $V_\pm$, takes the general form
\begin{equation}
	\label{eq:viscous_marker_real_space}
	\begin{split}
		\mathcal V_{\!*}(\br) &= \frac{\hbar}{8a_0^2}\Im\braket{\br|P\mathcal Q_0 Q \mathcal Q_1P|\br} \\
		&=\overline\eta_H+\frac{\hbar}{8\pi\ell_B^4}\mathcal C\,\rho^a g_{ab}\, \rho^b.
	\end{split}
\end{equation}
This is the second main result of our analysis. The parabolic term originates from the non-uniform Hall current induced in the bulk by the spatially varying electric field associated with the probe Eq.~\eqref{eq:rotating_saddle}, while the first term arises from the viscous stress in the bulk as described by Eq.~\eqref{eq:stress_tensor}. Notably, the local contribution to the transition rates is non-zero even where the electric field of the probe, $\bf E_\pm \propto -\boldsymbol{\nabla}V_\pm$, vanishes: 
\begin{equation}
	\label{eq:marker_at_origin}
	\mathcal V_{\!*}(\br=\boldsymbol{0}) = \overline\eta_H.
\end{equation}
The reason can be traced to the fact that the non-vanishing gradient of $\bf E_\pm$ generates a non-vanishing strain, which in turn produces a Hall viscous force.

Figure~\ref{fig:viscosity_real_space_marker} demonstrates that the real-space marker Eq.~\eqref{eq:viscous_marker_real_space} provides a robust probe of the Landau-orbit Hall viscosity, even away from the continuum limit. Panel (a) shows the marker is non-uniform in the bulk, reflecting spatially varying excitation rates. Panel (b) confirms that a slice of the marker has a parabolic profile with a non-zero offset, in agreement with Eq.~\eqref{eq:viscous_marker_real_space}. Panel (c) shows the marker at the center,
$\mathcal V_{\!*}(\br=\boldsymbol{0})$, as a function of flux per plaquette $\alpha$; it faithfully recovers the Landau-orbit viscosity via Eq.~\eqref{eq:marker_at_origin} up to $\alpha\simeq 1/5$, beyond which deviations accumulate.
Finally, panel (d) shows $\mathcal V_{\!*}(\br=\boldsymbol{0})$ as a function of the number of fermions $N$: for a partially filled band the marker fluctuates erratically, whereas at integer filling it correctly recovers the Landau-orbit viscosity, with only small corrections emerging at large $\mathcal C$.

{\it Frequency resolved measurement---}
We briefly comment on the possibility of isolating the viscous contribution ($\overline\eta_H$) from the quadratic Hall-type term in Eq.~\eqref{eq:viscous_marker_real_space} in the low-flux (Landau-level) limit. This can be achieved by frequency selection: the viscous contribution arises from transitions induced by $G_{0}\pm i G_1$, operators that are quadratic in the kinetic momenta and therefore  couple Landau levels differing by two cyclotron quanta, producing excitations at energies $\sim 2\hbar\omega_c$.
The Hall-type contribution, by contrast, arises from operators linear in the kinetic momenta, which  connect levels separated by a single cyclotron quantum, corresponding to excitation energies $\sim \hbar\omega_c$.

In the low-flux limit, lattice corrections weakly mix Landau levels while preserving their parity. Transitions split into two frequency sectors: the Hall-type response connects states of opposite parity, involving transitions whose energies differ by an odd multiple of $\hbar\omega_c$, whereas the viscous response connects states of the same parity, corresponding to energy differences equal to an even multiple of $\hbar\omega_c$~\cite{supp_mat}. The two responses thus appear as distinct spectral features. By suitably modifying the projector $Q$ in Eq.~\eqref{eq:viscous_marker_real_space}, one can selectively retain transitions within a chosen frequency window -- physically corresponding to filtering excitation rates by frequency.


We consider in Fig.~\ref{fig:frequency_selection} 
the case in which only the lowest band is occupied, and demonstrate how monitoring the excitation rates as a function of probe frequency isolates the constant viscous contribution [panel (a)] from the parabolic Hall contribution [panel (b)], providing a practical route to measuring the two independently.

\begin{figure}[t]
	\includegraphics[width=1\linewidth]{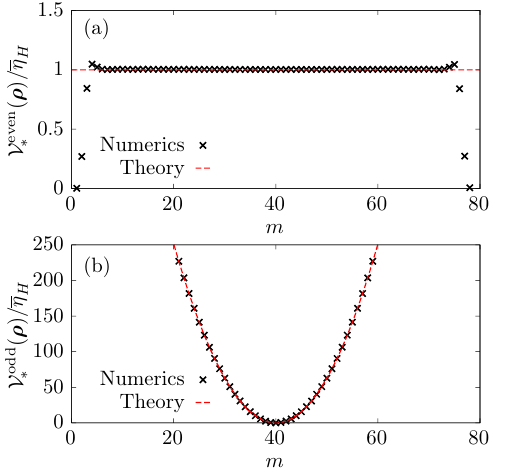}
	\caption{
		The viscosity-type (a) and Hall-type (b) contributions to the marker $\mathcal V_{\!*}(\br)$ in Eq.~\eqref{eq:viscous_marker_real_space} were isolated by frequency selecting the  even/odd transitions, respectively (see text).
		A slice of the obtained ``reduced" markers at $n\!=\!40$ is shown in both cases, and compared with the theoretical expectation [Eq.~\eqref{eq:viscous_marker_real_space}].
		Here, $N=600$ fermions occupy the lowest Hofstadter band, for a $80\times80$ lattice with  $J_x/J_y=1$, at flux $\alpha=1/10$.
	}
	\label{fig:frequency_selection}
\end{figure} 

{\it Conclusions---} We have introduced a spectroscopic framework to access Hall viscosity based on metric-sensitive circular dichroic probes. By engineering chiral perturbations that effectively ``shake" the underlying geometry -- either through direct modulation of the Landau-orbit metric or via a rotating saddle potential -- we showed that the differential absorption of opposite chiral drives provides direct access to the Berry curvature associated with area-preserving deformations, and hence to the Hall viscosity. This approach naturally leads to a local formulation, enabling spatially resolved markers that remain robust in lattice systems and beyond the strict continuum limit.

Our proposal is particularly well suited to cold-atom and quantum simulation platforms, where synthetic gauge fields and time-dependent control of lattice parameters can be used to implement both metric modulations and rotating quadrupolar potentials. In continuum-like setups~\cite{fletcher2021geometric,crepel2023geometric,mukherjee2022crystallization,schine2016synthetic,lunt2024realization}, these drives can be realized through controlled deformations of trapping potentials, while in lattice systems, they can be engineered via time-dependent hopping anisotropies or site-resolved modulations, as discussed for the Harper–Hofstadter model~\cite{aidelsburger2015measuring,leonard2023realization,impertro2024local,impertro2025strongly}. The spectroscopic nature of the protocol and its compatibility with frequency-resolved measurements make it readily accessible to current experimental techniques.

Extending this framework to interacting systems — in particular to fractional quantum Hall regimes — would provide a direct route to probing intrinsic orbital spin and distinguishing between competing topological orders, including non-Abelian phases~\cite{read:rezayi:2011,Zhu_interface}. It is interesting to contrast our single-particle probe with recent proposals accessing the geometric response through collective spin-2 ``graviton'' modes, which rely on more sophisticated modulations of interaction terms~\cite{xavier2025chiral,bacciconi2026chiralgravitonmodesnonabelian}; establishing connections between these approaches could offer complementary insights into the quantum geometry of strongly correlated states. Indeed, while the rotating-saddle probe in Eq.~\eqref{eq:rotating_saddle} gives access to the Landau-orbit Hall viscosity in non-interacting Chern-insulators [Eq.~\eqref{eq:marker_at_origin}], it also couples to the guiding-center degrees of freedom, opening a route to the full Hall viscosity of correlated insulators such as fractional quantum Hall states. More broadly, our results suggest that dichroic probes provide a versatile toolkit to explore geometric response functions in topological matter~\cite{sisterEllipticpaper}. Finally, it would be natural to extend these ideas to explore the Hall viscosity of higher-dimensional systems~\cite{Price_4D,Karabali_4D,Blagoje_4D,Karabali_2023,FQH_higher_dimension}, offering new opportunities to probe the interplay between topology, geometry, and dynamics in quantum systems.

{\it Acknowledgments---} 
We acknowledge enlightening discussions with Giandomenico Palumbo, Duncan Haldane, Cecile Repellin, Leonardo Mazza, Blagoje Oblak,  Iacopo Carusotto, Nehal Mittal, Zeno Bacciconi, Jie Wang, Laurens Vanderstraeten and Carolina Paiva. This research was financially supported by the ERC Grant LATIS, the FRS-FNRS (Belgium), the EOS project CHEQS, the Fondation ULB and the ANR PEPR
Grant QUTISYM ANR-23-PETQ-0002. B.~M. acknowledges support from the Security and Quantum Information Group (SQIG) in Instituto de Telecomunica\c{c}\~{o}es, Lisbon. This work is funded by FCT/MECI through national funds and when applicable co-funded EU funds under UID/50008: Instituto de Telecomunicações (IT). B.~M. further acknowledges the Scientific Employment Stimulus --- Individual Call (CEEC Individual) --- 2022.05522.CEECIND/CP1716/CT0001, with DOI: \href{https://doi.org/10.54499/2022.05522.CEECIND/CP1716/CT0001}{10.54499/2022.05522.CEECIND/CP1716/CT0001}. T.~O. acknowledges support from JSPS KAKENHI Grant Number JP24K00548, JST PRESTO Grant No. JPMJPR2353, and JST PRESTO Convergence Research Grant No. JPMJCR26XA.

\bibliography{bib.bib}

\clearpage
\cleardoublepage
\begin{supplementalMaterials}
	
	\title{Supplemental Material for\\ ``Hall viscosity from metric-sensitive dichroic probes''}
	
	\maketitle
	{
		\hypersetup{linkcolor=black}
		\tableofcontents
	}
	
	\section{Magnetic translations}
	We here introduce magnetic translation operators; we employ them to write an equivalent formulation of the Harper-Hofstadter Hamiltonian on a generic lattice, which easily allows to take its continuum limit.
	
	Using the kinetic momenta $\pi_i=p_i-eA_i$, which obey $[\pi_a,\pi_b]=i \hbar e B \epsilon_{ab}$,
	we can introduce magnetic translation operators
	\begin{equation}
		\label{eq:magnetic_translation}
		T_{(a_x,a_y)}=\exp(i (a_x \pi_x+a_y \pi_y)/\hbar);
	\end{equation}
	As a consequence of the Baker-Campbell-Hausdorff relation it follows that
	\begin{equation}
		\label{eq:magnetic_translation_commutator}
		T_{\mathbf d_1}T_{\mathbf{d}_2} =  T_{\mathbf d_1+\mathbf d_2} \exp\left(-i\,\frac{(\mathbf d_1 \times \mathbf d_2)\cdot\hat{z}}{2\ell_B^2}\right).
	\end{equation}
	
	\section{The Harper-Hofstadter model on a square lattice and its continuum limit}
	Consider the eigenstates $\ket{\psi_i}$ of the tight-binding Harper-Hofstadter Hamiltonian $H_{\rm HH}$ on a rectangular lattice with lattice spacings $a_x$ and $a_y$; we decompose them over the basis given by the lattice sites $\br=(m a_x,n a_y)$ as
	$\ket{\psi_i}=\sum_\br C_{\br,i}\ket{\br}$
	\begin{equation}
		\begin{split}
			E_i \ket{\psi_i}&= H_{\rm HH} \ket{\psi_i} \\&=
			\sum_\br C_{\br,i} \Biggl(-J_x \Bigl(\ket{\br+a_x\hat x}+\ket{\br-a_x\hat x}\Bigr)\\&\hspace{.5cm}-J_y\Bigl(e^{i2\pi\alpha m}\ket{\br+a_y \hat y}+e^{-i2\pi\alpha m}\ket{\br-a_y\hat y}\Bigr)\Biggr),
		\end{split}
	\end{equation}
	where $2\pi\alpha=a_x a_y/\ell_B^2$ is the magnetic flux per plaquette.
	
	The previous eigenvalue equation can be equivalently written as
	\begin{equation}
		\begin{split}
			E_i C_{\br,i} =& -J_x \Bigl(C_{\br-a_x\hat x,i}+C_{\br+a_x\hat x,i}\Bigr)\\&-J_y\Bigl(e^{i2\pi\alpha m}C_{\br-a_y \hat y,i}+e^{-i2\pi\alpha m}C_{\br+a_y \hat y,i}\Bigr).
		\end{split}
	\end{equation}
	If we promote the (discrete) probability amplitudes $C_{\br,i}$ to a family of continuous functions $C_{i}(\br)$, we can rewrite the right-hand side of the previous equation using the magnetic-translation operators
	\begin{subequations}
		\begin{equation}
			\label{eq:hh_magnetic_translations}
			H=-J_x\left(T_{(a_x,0)}^{}+T_{(a_x,0)}^\dagger\right)-J_y\left(T_{(0,a_y)}^{}+T_{(0,a_y)}^\dagger\right);
		\end{equation}   
		this form of the Hamiltonian, whose associated Schr\"odinger's equation reads
		\begin{equation}
			E_i C_{i}(\br)=H C_i(\br),
		\end{equation}
	\end{subequations}
	is very useful to analyze the low-energy continuum limit of the model, under the assumptions that the amplitudes $C_{\br,i}$ vary smoothly over the lattice spacing $a_{x(y)}$.
	Notice how, in order for the normalization to be correct when taking the continuum limit, one needs
	\begin{equation}
		\label{eq:normalization}
		C_{\br,i} = \sqrt{a_xa_y}\, C_i(\br):
	\end{equation}
	this way $1=\sum_{\br}|C_{\br,i}|^2=\sum_{\br}|C_i(\br)|^2 a_x a_y \rightarrow \int d^2\br|C_{\br,i} |^2$.

	Before we do so, we stress how, if we define a rectangular lattice through
	\begin{equation}
		\label{eq:hh_gauge_choice}
		\ket{m,n}=T_{(a_x,0)}^m T_{(0,a_y)}^n\ket{0,0}
	\end{equation} 
	where $x= a_x m$, $y=a_y n$ are the lattice points, we recover the same Hamiltonian after projecting Eq.~\eqref{eq:hh_magnetic_translations} onto these lattice sites.
	\begin{subequations}
		Indeed, acting with $T_{(a_x,0)}$ onto $\ket{m,n}$ gives
		\begin{equation}
			\begin{split}
				T_{(a_x,0)}\ket{m,n}=&\ket{m+1,n}    
			\end{split}
		\end{equation}
		while, using Eq.~\eqref{eq:magnetic_translation_commutator}, the action of $T_{(0,a_y)}$ onto $\ket{m,n}$ yields
		\begin{equation}
			\begin{split}
				T_{(0,a_y)}\ket{m,n}&=e^{i2\pi\alpha m}\ket{m,n+1}.
			\end{split}
		\end{equation}	
	\end{subequations}
	Notice how these identities directly lead to the Harper-Hofstadter Hamiltonian Eq.~\eqref{eq:hh_hamiltonian}
	\begin{equation}
		\label{eq:hh_hamiltonian_sm}
		\begin{split}
			H_{\rm HH}&=\sum_{\mathbf{r},\mathbf{r}'}\braket{\mathbf{r}'|H|\mathbf{r}}  c_{\mathbf{r}'}^\dagger c_{\mathbf{r}}=\\&-J_x \sum_{m,n} c_{m+1,n}^\dagger c_{m,n} -J_y \sum_{m,n} e^{i2\pi\alpha m}c_{m,n+1}^\dagger c_{m,n}+\hc
		\end{split}
	\end{equation}	
	Notice finally how Eq.~\eqref{eq:hh_gauge_choice} is a gauge choice: it reflects onto the Landau gauge of Eq.~\eqref{eq:hh_hamiltonian_sm}.
	
	It is now possible to analyze the continuum limit of the Harper-Hofstadter model;
	indeed, when the lattice spacing is much smaller than the magnetic length $a_{x(y)}\ll \ell_B$ one expects, at low-energies, the amplitudes $C_{\br,i}$ to vary smoothly over the lattice spacing $a_{x(y)}$; 
	we then perform an ordered expansion of the magnetic-translation operators appearing in Eq.~\eqref{eq:hh_magnetic_translations},
	retaining only the lowest order terms.

	When $a_{x(y)}/\ell_B\ll 1$ we get
	\begin{equation}
		\label{eq:hh_limit_sm}
		H \simeq -2(J_x+J_y) + \frac{1}{2\hbar^2}\Bigl(2J_x a_x^2\pi_x^2 + 2J_y a_y^2 \pi_y^2\Bigr)
	\end{equation}
	up to $\mathcal O(a_{x(y)}^4/\ell_B^4)$ corrections.
	This is exactly (up to the constant energy shift $-2(J_x+J_y)$) the result Eq.~\eqref{eq:ll_limit} quoted in the main text upon focusing on the square-lattice case $a_x=a_y\equiv a_0$.
	
	\section{The Harper-Hofstadter model on a triangular lattice and its continuum limit}
	In this section we focus instead on a triangular lattice with primitive vectors
	$\mathbf{d}_1= (a_0,0)$ and $\mathbf{d}_2=a_0(\cos(\theta),\sin(\theta))$,
	pierced by a uniform magnetic field.
	We will focus mostly on the equilateral case, with $\theta=\pi/3$.
	
	Following the previous section, we can write the Harper-Hofstadter Hamiltonian on such a lattice using magnetic-translation operators as
	\begin{equation}
		\label{eq:triangular_hh}
		\begin{split}
			H = -&t_1 (T_{\mathbf{d}_1}^{}+T_{\mathbf{d}_1}^{\dagger}) -t_2 (T_{\mathbf{d}_2}^{}+T_{\mathbf{d}_2}^{\dagger}) -\\ - &t_3 (T_{\mathbf{d}_2-\mathbf{d}_1}^{}+T_{\mathbf{d}_2-\mathbf{d}_1}^\dagger);
		\end{split}
	\end{equation}
	this form is again useful since it allows to study the low-energy continuum limit of the lattice model, which we are now going to derive by projecting Eq.~\eqref{eq:triangular_hh} onto the lattice sites
	\begin{equation}
		\label{eq:triangular_lattice_gauge}
		\ket{m,n}=T_{\mathbf{d}_1}^m T_{\mathbf{d}_2}^n\ket{0,0}.
	\end{equation}	
	
	Using Eq.~\eqref{eq:magnetic_translation_commutator} we get $T_{\mathbf d_2-\mathbf d_1} = T_{\mathbf d_2} T_{-\mathbf d_1} e^{i 2\pi \alpha}$
	where the flux per plaquette in the triangular case is
	\begin{equation}
		2\pi\alpha= \frac{1}{2} a_0^2 \sin(\theta) B.
	\end{equation}
	Notice how this relation is indeed consistent with the Aharonov-Bohm condition $T_{-\mathbf{d}_2}T_{\mathbf{d}_2-\mathbf d_1}T_{\mathbf{d}_1}= e^{i2\pi\alpha}$.
	A second useful identity which can be derived using Eq.~\eqref{eq:magnetic_translation_commutator} is $T_{\mathbf d_2}T_{\mathbf d_1} = T_{\mathbf d_1}T_{\mathbf d_2} e^{i 4\pi\alpha}$, which is consistent with the Aharonov-Bohm flux picked up when traversing two triangles.
	
	We can finally understand the action of the relevant magnetic translations on the lattice states defined in Eq.~\eqref{eq:triangular_lattice_gauge}.
	These produce the following hopping rules
	\begin{subequations}
		\begin{align}
			T_{\mathbf{d}_1} \ket{m,n} &= \ket{m+1,n}\\
			T_{\mathbf{d}_2} \ket{m,n} &= e^{i 4\pi m\alpha} \ket{m,n+1}\\
			T_{\mathbf{d}_2-\mathbf d_1} \ket{m,n} &= 
			e^{i 2\pi(2m-1)\alpha} \ket{m-1,n+1}.
		\end{align}
	\end{subequations}
	As a consequence, the continuum Hamiltonian Eq.~\eqref{eq:triangular_hh}, projected onto the lattice sites Eq.~\eqref{eq:triangular_lattice_gauge}, reads
	\begin{equation}
		\begin{split}
			H = - \sum_{m,n} t_1& c_{m+1,n}^\dagger c_{m,n} + t_2 e^{i 4\pi m \alpha} c_{m,n+1}^\dagger c_{m,n} +\\+&\, t_3 e^{i2\pi (2m-1)\alpha} c_{m-1,n+1}^\dagger c_{m,n} + \hc     
		\end{split}
	\end{equation}
	which was previously studied for example in~\cite{Kazusumi_2006}.
	
	Analogously to what was done in the previous section, we here take the continuum limit of Eq.~\eqref{eq:triangular_hh} and study which metric emerges.
	Neglecting the constant energy offset we get, at quadratic order,
	\begin{equation}
		H \simeq \frac{t_1}{\hbar^2} (d_1^\alpha\pi_\alpha)^2 + \frac{t_2}{\hbar^2} (d_2^\alpha \pi_\alpha)^2 + \frac{t_3}{\hbar^2} (d_2^\alpha \pi_\alpha + d_1^\alpha\pi_\alpha)^2;
	\end{equation}	
	simplifying everything leads to
	\begin{equation}
		H \simeq \frac{1}{2\meff}\left(g^{xx} \pi_x^2 + g^{xy}(\pi_x \pi_y+\pi_y \pi_x) + g^{yy}\pi_y^2\right)
	\end{equation}	
	where $\meff^{-1}=2a_0^2\sqrt{(t_1t_2+t_1t_3+t_2t_3)\sin^2(\theta)}/\hbar^2$ and
	\begin{subequations}
		\begin{align}
			\frac{\hbar^2/a_0^2}{2\meff} g^{xx} &= t_1+t_3-2t_3 \cos(\theta) +(t_2+t_3)\cos(\theta)^2\\
			\frac{\hbar^2/a_0^2}{2\meff}g^{xy} &=( (t_2+t_3)\cos(\theta)-t_3)\sin(\theta)	\\
			\frac{\hbar^2/a_0^2}{2\meff}g^{yy} &= (t_2+t_3)\sin(\theta)^2.
		\end{align} 
	\end{subequations}
	When $\theta=\frac \pi 3$ the lattice is an equilateral triangular; if, in this limit, we also have the same hopping amplitudes, then $t_1=t_2=t_3$ and the inverse mass simplifies to $\meff^{-1}=3a_0^2t/\hbar^2$.
	As a consequence, $g^{xx}=g^{yy}=1$ and $g^{xy}=0$: the metric is Euclidean.
	Off-diagonal elements can be engineered by introducing small asymmetries in the hopping rates ($t_1$, $t_2$, $t_3$) and/or in the lattice geometry (regulated by $\theta$): this can be used to engineer the metric deformation Eq.~\eqref{eq:metric_deformation_viscosity_probe}. This is what we are going to discuss in the next section.
	
	\section{Physical implementation of the metric shaking protocol on a triangular lattice}
	We here consider the simplest case of an optical lattice with $\theta=\frac{\pi}{3}$; 
	such a lattice can be created by superposing three laser beams with the same frequency $\omega_L$, but suitably chosen wavevectors and field amplitudes $E_i$~\cite{Becker_2010,Yamamoto_2020}. 
	In this case, the hopping amplitudes can be independently tuned:
	one has $t_1\propto E_2E_3$, $t_2\propto E_1E_3$ and $t_3\propto E_1E_2$.
	We will focus on the case where $t_1\simeq t_2\simeq t_3$, so that the unperturbed metric is the Euclidean one.
	
	To achieve the metric perturbation Eq.~\eqref{eq:metric_deformation_viscosity_probe}, the hopping amplitudes can be independently tuned by modifying the field amplitudes in a time-dependent manner on a time-scale much longer than $\omega_L^{-1}$. 
	If we expand $E_i \simeq E+\delta E_i(t)$, we can write
	\begin{subequations}
		\begin{align}
			t_1'&\simeq t_0 + w(\delta E_2+\delta E_3)\\
			t_2'&\simeq t_0 + w(\delta E_1+\delta E_3)\\
			t_3'&\simeq t_0 + w(\delta E_1+\delta E_2),
		\end{align}
	\end{subequations}
	with $w$ the appropriate conversion factor.
	The couplings need to be engineered in such a way that the effective mass $\meff^{-2}\propto E_1 E_2 E_3 (E_1+E_2+E_3)$ remains constant to linear order, up to quadratic $\mathcal{O}(\delta E_i\delta E_j)$ terms: $\meff^{-1}\simeq 3a_0^2t_0/\hbar^2$.
	This can be seen to be the case if $\sum_i \delta E_i = 0$.
	
	Using this constraint, the metric deformation on top of the unperturbed one $g^{ab}=\delta^{ab}$ can be written in components as
	\begin{subequations}
		\begin{align}
			\delta g^{xx} &= -\delta g^{yy} = -\frac{w}{2t_0}\delta E_1 \\
			\delta g^{xy} &=  -\frac{w}{2t_0\sqrt{3}}(\delta E_1+2\delta E_2)
		\end{align}
	\end{subequations}
	which can immediately be seen to preserve the unimodularity condition at linear order.
	Choosing
	\begin{subequations}
		\begin{align}
			\delta E_1 &= -\frac{2t_0}{w}\delta \cos(\omega t)\\
			\delta E_2 &= +\frac{2t_0}{w} \delta \cos\left(\omega t \mp \frac{\pi}{3}\right),
		\end{align}
		where we recall that $\delta = 2\meff u_0\ell_B^2/\hbar^2$, we achieve the desired metric perturbation Eq.~\eqref{eq:metric_deformation_viscosity_probe} on top of a Euclidean metric.
		The last field amplitude,  $\delta E_3$, is determined by the unimodularity constraint $\delta E_3=-(\delta E_1+\delta E_2)$ as discussed above. It simplifies to
		\begin{equation}
			\delta E_3 = +\frac{2t_0}{w}\delta \cos\left(\omega t\pm \frac\pi 3\right).
		\end{equation}
	\end{subequations}

	\section{A local marker for the viscosity}
	We here show how the local marker expression Eq.~\eqref{eq:viscosity_marker} can be derived from the integrated rates Eq.~\eqref{eq:fgr_rates}.
	
	Notice how $G_0$, $G_1$ are one-body operators. As such, on a non-intercting fermionic system, we can express the summation over all the transition matrix elements,
	$\sum_{n\neq0} |\braket{\Psi_n|G_0\mp G_1|\Psi_0}|^2$, 
	as a summation over the occupied and empty states.
	The frequency-integrated transition rates become
	\begin{equation}
		\label{eq:free_fermions_integrated_rates}
		\Gamma_\pm^{\rm int}=2\pi \frac{|u_0|^2}{4\hbar^2}\sum_{\alpha \in {\rm occ.}}  \sum_{\beta \in {\rm emp.}} |\braket{\beta|G_0\mp i G_1|\alpha}|^2.
	\end{equation}
	We introduce projectors over the occupied and empty states
	\begin{subequations}
		\begin{align}
			Q &= \sum_{\beta\in {\rm emp.}} \ket{\beta}\bra{\beta}
			\\
			P = 1-Q &=  \sum_{\alpha\in {\rm occ.}} \ket{\alpha}\bra{\alpha}
		\end{align}    
	\end{subequations}
	which allows us to rewrite Eq.~\eqref{eq:free_fermions_integrated_rates} as a trace
	\begin{equation}
		\begin{split}
			\Gamma_\pm^{\rm int}=&2\pi \frac{|u_0|^2}{4\hbar^2}\sum_{\alpha} \braket{\alpha|P(G_0\pm i G_1)Q(G_0\mp i G_1)P|\alpha} \\=& 2\pi \frac{|u_0|^2}{4\hbar^2} {\rm Tr}\Bigl[P(G_0\pm i G_1)Q(G_0\mp i G_1)P\Bigr].
		\end{split}
	\end{equation}
	Taking the difference $\Delta\Gamma^{\rm int}=\Gamma_+^{\rm int}-\Gamma_-^{\rm int}$ we can express the dichroic signal as
	\begin{equation}
		\begin{split}
			\Delta\Gamma^{\rm int}=2\pi |u_0|^2 \Im \Bigl[{\rm Tr}\bigl( PG_0QG_1P  \bigr)\Bigr].
		\end{split}
	\end{equation}
	We use use the lattice vectors $\ket{m,n} \equiv \ket{\boldsymbol{\rho}}$ to perform the trace:
	\begin{equation}
		\begin{split}
			\Delta\Gamma^{\rm int}=2\pi \frac{|u_0|^2}{\hbar^2} \sum_\br\Im \Braket{\br| PG_0QG_1P  |\br}
		\end{split}
	\end{equation}
	which straightforwardly leads to the viscosity marker introduced in Eq.~\eqref{eq:viscosity_marker} when identifying 
	\begin{subequations}
		\begin{align}
			\Delta\Gamma^{\rm int} &= 16\pi \hbar^{-3} |u_0|^2 \sum_\br \mathcal V(\br) a_0^2
			\\
			\mathcal V(\br) &= \frac{\hbar}{8a_0^2}\Im \Braket{\br| PG_0QG_1P  |\br}.
		\end{align}    
	\end{subequations}
	
	\begin{figure}[t]
		\includegraphics[width=1\linewidth]{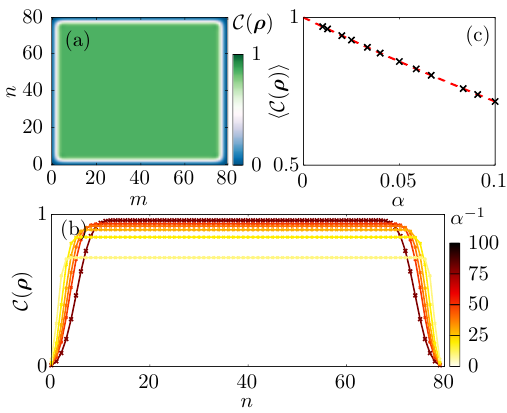}
		\caption{
			The Chern number is extracted using the Chern marker defined in Eq.~\eqref{eq:bulk_chern_marker_continuum} in the case of a non-interacting fermionic Harper-Hofstadter lattice with $N_x=N_y=80$ lattice sites in the isotropic case, $J_x=J_y=1$.
			(a) The Chern marker for the Harper-Hofstadter case for $N=280$ fermions filling up the lowest band at $\alpha=1/20$; a slice of the marker at constant $m=N_x/2$ is compared to the theoretical value Eq.~\eqref{eq:uniform_marker_value} in (b), for various values of $\alpha$.
			(c) Value of the marker averaged over a $10\times10$ square around $(N_x/2,N_y/2)$ (black crosses), as a function of the magnetic flux per plaquette $\alpha$, showing how it extrapolates to the Landau-orbit viscosity as $\alpha\rightarrow0$. The points have been fitted with a quadratic function of $\alpha$ (red dashes); as $\alpha\rightarrow0$, we find $\braket{\mathcal C(\br)}=0.9997(6)$.
		}
		\label{fig:chern_marker}
	\end{figure}

	\section{Chern marker decomposition}
	The Chern marker 
	\begin{equation}
		\mathfrak{C}(\br) = 4\pi\, \Im \bra{\br}PxQyP\ket{\br}
	\end{equation}
	is well-known to be a robust indicator of a non-interacting insulator's Chern number, even if $a_0/\ell_B$ is non-zero.
	In this section we consider lattice models which reduce to Landau levels in the continuum limit $\ell_B\gg a_0$, and analyze how the Chern marker decomposition in terms of guiding centers and kinetic momenta carries over to non-zero $a_0/\ell_B$.
	We will here see how, while in the continuum limit the bulk value of the Chern marker is entirely characterized by the kinetic momenta
	\begin{equation}
		\label{eq:bulk_chern_marker_continuum}
		\mathfrak{C}(\br) \,\xrightarrow[a_0\ll\ell_B] \,\mathcal C(\br) = -4\pi\,\frac{\ell_B^4}{\hbar^2}\, \Im \bra{\br}P\pi_yQ\pi_xP\ket{\br},
	\end{equation}
	when $a_0/\ell_B$ is non-zero deviations with respect to the expected bulk value occur.
	
	In particular, Eq.~\eqref{eq:bulk_chern_marker_continuum} is a consequence of the decomposition $r^i=R^i-\frac{\ell_B^2}{\hbar}\epsilon^{ij}\pi_j$ of the coordinates in terms of guiding centers and kinetic momenta: recalling that $Q(P)$ is a projector onto the empty (occupied) states, in the bulk of a non-interacting system only the $\pi$s can contribute since the 
	transitions induced by the
	$R$s are forbidden by Pauli exclusion: they indeed occur from an occupied state to another occupied one (we here skip the detailed analysis, since it parallels the one carried over in the next section). 
	
	The numerical results for Eq.~\eqref{eq:bulk_chern_marker_continuum} in the case of a fully-filled lowest band of the Harper-Hofstadter lattice Eq.~\eqref{eq:hh_hamiltonian} are shown in Fig.~\ref{fig:chern_marker}.
	It can be seen from panel $(a)$ that the marker $\mathcal C(\br)$ is uniform in the bulk; panels $(b)$ and $(c)$ show how the marker approaches the expected value when $2\pi\alpha=(a_0/\ell_B)^2$ decreases, but also demonstrate how the lattice discretization of the kinetic momenta $\pi_a$, now explicitly appearing in the marker Eq.~\eqref{eq:bulk_chern_marker_continuum}, introduces deviations with respect to the expected value $C_{\rm MB}=1$.
	
	\section{Rotating-saddle -- marker analysis}
	\label{smsec:marker_analysis}
	We here give two different derivations of Eq.~\eqref{eq:viscous_marker_real_space}.
	The first one is its direct calculation in the low-flux limit, 
	while the second one relies on the calculation of the power absorbed by a quantum Hall sample when subjected to a non-uniform electric field.
	
	\subsection{Landau levels calculation}
	We here perform the explicit calculation in the low-flux limit. To do so, we begin by introducing creation-annihilation operators which diagonalize the Hamiltonian in the continuum limit, Eq.~\eqref{eq:ll_limit}.
	\begin{subequations}
		We can decompose both the cyclotron-motion and guiding-center operators in terms of bosonic variables. In particular, for the former we write
		\begin{equation}
			\begin{split}
				\Pi_x &= \frac{\hbar}{\ell_B} \frac{a+a^\dagger}{\sqrt{2}}
				\\
				\Pi_y &= \frac{\hbar}{\ell_B} \frac{a-a^\dagger}{\sqrt{2}\,i}
			\end{split}
		\end{equation}
		while for the latter
		\begin{equation}
			\begin{split}
				\mR^x &= \ell_B\frac{b+b^\dagger}{\sqrt{2}}
				\\
				\mR^y &= -\ell_B\frac{b-b^\dagger}{\sqrt{2}\,i}.
			\end{split}
		\end{equation}
	\end{subequations}
	As anticipated, the $a$ and $b$ operators satisfy bosonic commutation relations, $[a,a^\dagger]=[b,b^\dagger]=1$.
	The single particle states can be labelled by the two associated quantum numbers $\ket{n,m}$: $n$ being the Landau-level index and $m$ the guiding center one.
	Since we focus on bulk properties, we consider an infinite non-interacting fermionic system consisting of $\mathcal C$ fully filled Landau levels 
	\begin{equation}
		\label{eq:iqh_ground_state} 
		\ket{\Psi_0}= \mathcal A\bigotimes_{n=0}^{\mathcal C-1} \bigotimes_{m=0}^\infty \ket{n,m},
	\end{equation}
	where $\mathcal A$ denotes the antisymmetrization operator.

	We then use the coordinate splitting $X^i=\mR^i-\frac{\ell_B^2}{\hbar}\epsilon^{ij}\Pi_j $ 
	to write
	\begin{equation}
		\label{eq:saddle_decomposition}
		\begin{split}
			\frac{(X^1)^2-(X^2)^2}{\ell_B^2} &= \mathcal D_1+\mathcal D_2 + \mathcal D_3
			\\
			\frac{2X^1 X^2}{\ell_B^2} &= \mathcal E_1 + \mathcal E_2 + \mathcal E_3,
		\end{split}
	\end{equation}
	where the operators $\mathcal D_i$ and $\mathcal E_j$ are listed in Tab.~\ref{tab:operators_list} together with all the possible quantum number changes.
	In the marker Eq.~\eqref{eq:viscous_marker_real_space}, every possible combination appears:
	\begin{subequations}
		\label{eq:marker_real_space_contributions}
		\begin{equation}
			\mathcal V_{\!*}(\br) = \sum_{ij}\frac{\hbar}{8a_0^2} \Im\, \mathcal K_{ij}
		\end{equation}
		where
		\begin{equation}
			\mathcal K_{ij} = P \mathcal D_i Q \mathcal E_j P.
		\end{equation}    
	\end{subequations}
	Notice how $\mathcal E$ must take an occupied state to an empty state, while $\mathcal D$ needs to take that same empty state to a second occupied state.
	
	\renewcommand{\arraystretch}{1.3}
	\begin{table}[h!]
		\centering
		\begin{tabular}{c|c|c|c|c}
			\hline
			\multicolumn{3}{c|}{Operator} & $\Delta n$ & $\Delta m$ \\
			\cline{2-3}
			\hline
			$\mathcal D_1$ & $- \frac{\ell_B^2}{\hbar^2} (\Pi_x^2-\Pi_y^2)$ &  $-(a^2+(a^\dagger)^2)$ & $\pm 2$ & $0$ \\
			$\mathcal D_2$ & $- \frac{2}{\hbar}(\mR^x\Pi_y+\mR^y\Pi_x)$ & $2i(a b^\dagger-a^\dagger b)$ & $ \pm 1$ & $ \pm 1$ \\
			$\mathcal D_3$ & $(\mR_x^2-\mR_y^2)/\ell_B^2$ & $b^2+(b^\dagger)^2 $ & $0$ & $\pm 2$ \\
			\hline
			$\mathcal E_1$ & $- \frac{\ell_B^2}{\hbar^2} (\Pi_x\Pi_y+\Pi_y\Pi_x)$ & $i(a^2-(a^\dagger)^2)$ & $\pm 2$ & $0$ \\
			$\mathcal E_2$ & $\frac{2}{\hbar}(\mR^x\Pi_x-\mR^y\Pi_y)$ & $2(a b^\dagger+a^\dagger b)$ & $\pm 1$ & $\pm 1$ \\
			$\mathcal E_3$ & $(\mR^x \mR^y+\mR^y \mR^x)/\ell_B^2$ & $i(b^2-(b^\dagger)^2)$ & $0$ & $\pm 2$ \\
			\hline
		\end{tabular}
		\caption{Operators appearing in the real-space viscosity marker in the continuum limit with the corresponding quantum number changes.}
		\label{tab:operators_list}
	\end{table}
	\renewcommand{\arraystretch}{1}
	There are $9$ such combinations. However, notice that all the terms involving a  quadratic form in the guiding-centers, i.e. the terms with either $i=3$ or $j=3$, produce transitions with $|\Delta m|=2$, $\Delta n=0$: they cannot connect an occupied state to an empty one due to the ground-state structure Eq.~\eqref{eq:iqh_ground_state} (if not at the system's edge).
	There are therefore only $4$ combinations one needs to keep track of.
	
	\subsubsection{The viscous term: $\mathcal K_{11}$}
	Since we assume all the levels below $\mathcal C$ (excluded) are occupied, $\mathcal K_{11}$ can be simplified to
	\begin{equation}
		\mathcal K_{11} = i \sum_{n=0}^{\mathcal C-1} \sum_{m=0}^{\infty} \ket{n,m}\bra{n,m} \sum_{n'=\mathcal C}^\infty\, \delta_{n',n+2} (n+1)(n+2).
	\end{equation}
	Only two terms contribute: $(n,n') = (\mathcal C-2,\mathcal C)$ and $(n,n')=(\mathcal C-1,\mathcal C+1)$; physically, the contributions coming $\Delta n=\pm2$ transitions can only come from the highest two Landau levels.
	Furthermore, since the Landau level density is independent of the Landau level index,
	\begin{equation}
		\sum_{m=0}^{\infty} |\braket{\boldsymbol{\rho}|n,m}|^2 = \frac{a_0^2}{2\pi\ell_B^2},
	\end{equation}
	the contribution to $\mathcal V_{\!*}(\br)$ is
	\begin{equation}
		\label{eq:k11}
		\frac{\hbar}{8a_0^2}\Im \braket{\boldsymbol{\rho}| \mathcal K_{11}  |\boldsymbol{\rho}} = \frac{\hbar}{8\pi\ell_B^2}\mathcal C^2 = \overline \eta_H.
	\end{equation}
	As expected, from the term which is purely quadratic in the $\Pi$s we recover the Hall viscosity.
	
	\subsubsection{The Hall term: $\mathcal K_{22}$}
	Carrying out the same math, we obtain
	\begin{equation}
		\mathcal K_{22}= 4i\,\mathcal C\sum_{m=0}^{\infty}m\,\ket{\mathcal C-1,m}\bra{\mathcal C-1,m}. 
	\end{equation}
	The contribion to the marker is then given by
	\begin{equation}
		\label{eq:k22}
		\begin{split}
			\frac{\hbar}{8a_0^2}\Im \braket{\boldsymbol{\rho}| \mathcal K_{22}  |\boldsymbol{\rho}} &= \frac{\hbar}{8a_0^2}4\mathcal C \sum_{m=0}^{\infty} m |\!\braket{\boldsymbol{\rho}|\mathcal C-1,m}|^2 =\\& \frac{\hbar}{8\pi\ell_B^2}\mathcal C \Bigl(2(\mathcal C-1)+\frac{\rho^a g_{ab}\,\rho^b}{\ell_B^2}\Bigr).
		\end{split}
	\end{equation}
	Notice the presence of a $\propto \mathcal C(\mathcal C-1)$ constant offset: this gets canceled by the $\mathcal K_{12}$ and $\mathcal K_{21}$ terms.
	
	\subsubsection{The mixed terms: $\mathcal K_{12}$ and $\mathcal K_{21}$}
	In this case we obtain
	\begin{equation}
		\begin{split}
			\mathcal K_{12} &= 
			- 2\mathcal C \sqrt{\mathcal C-1}\sum_{m=1}^\infty\sqrt{m}\,\ket{\mathcal C-2,m-1}\bra{\mathcal C-1,m}  
		\end{split}
	\end{equation}
	and $\mathcal K_{21} = - \mathcal K_{12}^\dagger$.
	Performing the summations we get, for the marker contributions,
	\begin{equation}
		\label{eq:k12_k21}
		\frac{\hbar}{8a_0^2}\Im \braket{\boldsymbol{\rho}| \mathcal K_{12}  |\boldsymbol{\rho}} =     \frac{\hbar}{8a_0^2}\Im \braket{\boldsymbol{\rho}| \mathcal K_{21}  |\boldsymbol{\rho}} = - \frac{\hbar}{8\pi\ell_B^2} \mathcal C (\mathcal C-1).
	\end{equation}
	Notice how these two contributions, taken together, exactly cancel out the offset term in Eq.~\eqref{eq:k22}.
	
	\subsubsection{The final result}
	Summing according to Eq.~\eqref{eq:marker_real_space_contributions} the results of Eq.~\eqref{eq:k11}, ~\eqref{eq:k22}, ~\eqref{eq:k12_k21} we obtain the final expression for the marker
	\begin{equation}
		\mathcal V_{\!*}(\br) = \overline \eta_H +  \frac{\hbar}{8\pi \ell_B^4}\mathcal C\, \rho^a g_{ab}\,\rho^b
	\end{equation}
	which is the real-space marker result Eq.~\eqref{eq:viscous_marker_real_space} quoted in the main text.

	\subsection{Power absorption from an inhomogeneous probe}
	In the section we instead demonstrate Eq.~\eqref{eq:viscous_marker_real_space} from semiclassical arguments.
	For simplicity we work in the Euclidean case, assuming isotropy and rotational invariance.
	
	In this case the rotating saddle potential of Eq.~\eqref{eq:rotating_saddle} becomes
	\begin{equation}
		V_\pm= \Re\left[ v_0\, e^{-i \omega t} \frac{(x^2-y^2)\mp i 2xy}{\ell_B^2}\right].
	\end{equation}
	We can associate a complex electric field to this potential; it reads
	\begin{equation}
		\label{eq:probe_field}
		\begin{split}
			{\bf E}_\pm = -\frac{1}{e}{\boldsymbol{\nabla}} V_\pm = 2\frac{v_0}{e\ell_B^2} (x\mp i y) 
			{\boldsymbol{\epsilon}}_\mp
			e^{-i\omega t};
		\end{split}
	\end{equation}
	the “real” electric field is just its real part. Here we introduced a polarization vector ${\boldsymbol{\epsilon}}_\pm = 
	\begin{psmallmatrix}
		1\\
		\pm i
	\end{psmallmatrix}$.
	
	This field must be delivering the energy needed for the transitions occurring in the system. If $P_\pm(\omega)$ is the power absorbed at frequency $\omega$, the dichroic transition rates can be expressed as~\cite{BennettStern,goldman2024relating}
	\begin{equation}
		\Delta\Gamma^{\rm int} = \int_0^\infty\frac{\Delta P(\omega)}{\hbar\omega}d\omega,
	\end{equation}
	where $\Delta P(\omega)=P_+(\omega)-P_-(\omega)$.
	
	Our goal here is therefore to compute the power absorbed by the system at frequency $\omega$, time-averaged over a time-window $T=\frac{2\pi}{\omega}$ and integrated over the whole system. 
	There are two contributions:
	\begin{subequations}
		an electromagnetic one, due to Joule heating
		\begin{equation}
			\label{eq:joule_effect}
			\begin{split}
				P_{\rm Joule,\pm}(\omega) &= \int d^2{\bf r} \int \frac{dt}{T} \Re({\bf J})\cdot \Re({\bf E_{\pm}}) =
				\\ &
				\frac{1}{2} \int d^2{\bf r}\, \Re({\bf J}\cdot {\bf E}_{\pm}^*)
			\end{split}
		\end{equation}
		and a fluid-dynamics one, namely the work done by the internal stresses
		\begin{equation}
			\label{eq:stress_heating}
			\begin{split}
				P_{\rm stress,\pm}(\omega) &= -\int d^2{\bf r} \int \frac{dt}{T} \Re(\Sigma^{a}_b)\cdot \Re(\dot u_{a}^b) =
				\\ &
				-\frac{1}{2} \int d^2{\bf r}\, \Re(\Sigma^{a}_b\,\dot u_{a}^{b\,*})
			\end{split}
		\end{equation}
	\end{subequations}
	where 
	\begin{equation}
		\label{eq:stress_tensor_expansion}
		\Sigma^a_b({\bf r},t)=-\int d^2{\bf r'}\,\widetilde\eta^{aecf}({\bf r}-{\bf r}',t-t')\epsilon_{eb}\epsilon_{fd}\dot u^{d}_c({\bf r}',t')
	\end{equation}
	is the stress tensor~\cite{Park_2014} (cf. Eq.~\eqref{eq:stress_tensor}) and $\dot u^{a}_b$ the strain tensor variation
	\begin{equation}
		\label{eq:strain_variation}
		\dot u^{a}_b = \frac{1}{2}\left(\frac{\partial v^a}{\partial x^b}+\frac{\partial v^b}{\partial x^a}\right);
	\end{equation}
	the velocity field $v^a$ is here determined by the external potential through the guiding-center drift relation
	\begin{equation}
		v^a = - \frac{1}{B}\epsilon^{ab} E_{\pm,b}.
	\end{equation}
	
	\subsubsection{The electromagnetic contribution}
	The current ${\bf J}$ is set by the Ohmic linear response relation
	\begin{equation}
		\label{eq:hall_response}
		J_i(\boldsymbol{r},t) = \int d^2{\bf r}' \int dt' \sigma_{ij}({\bf r}-{\bf r}',t-t') E_j({\bf r}', t').
	\end{equation}
	The conductance tensor $\sigma_{ij}$ is the macroscopic information on the bulk's response, and is set by the particular topological order analyzed.
	Plugging the field Eq.~\eqref{eq:probe_field} and the current Eq.~\eqref{eq:hall_response} into the power formula Eq.~\eqref{eq:joule_effect} we get (for simplicity, and without loss of generality, we take $v_0$ to be real)
	\begin{equation}
		\label{eq:joule_power_tmp}
		\begin{split}
			P_{\rm Joule,\pm}(\omega) =
			2 \frac{v_0^2}{e^2\ell_B^4} &\int \frac{d^2{\bf q}}{4\pi^2}\int d^2{\bf r} \int d^2{\bf r}' \,\Re\Biggl[\epsilon_{\mp, i}^*\epsilon_{\mp, j}
			\\& (x'\mp i y')(x\pm i y) \sigma_{ij}({\bf q},\omega) e^{i {\bf q}\cdot ({\bf r}-{\bf r}')}
			\Biggr]		 
		\end{split}
	\end{equation}
	where $\sigma_{ij}({\bf q},\omega)$ is the Fourier transform of $\sigma_{ij}({\bf r},t)$
	\begin{equation}
		\sigma_{ij}({\bf q},\omega) = \int dt' \int d^2{\bf r} \,\sigma_{ij}({\bf r},t) e^{i\omega t} e^{-i {\bf q}\cdot ({\bf r}-{\bf r}')};
	\end{equation}
	thanks to the symmetry assumptions the conductance tensor is constrained as
	\begin{equation}
		\sigma = \begin{pmatrix}
			\sigma_{xx} & \sigma_{xy}\\
			- \sigma_{xy} & \sigma_{xx}
		\end{pmatrix};
	\end{equation}
	using this result, the contraction of the conductance tensor with the polarization vectors simplifies
	\begin{equation}
		\begin{split}
			\epsilon_{\mp, i}^*\epsilon_{\mp, j} \sigma_{ij}({\bf q},\omega) =& 2(\sigma_{xx}({\bf q},\omega)-i \sigma_{xy}({\bf q},\omega)) \equiv\\\equiv& 2F_H({\bf q},\omega).
		\end{split}
	\end{equation}
	We then write 
	\begin{equation}
		(x'\mp i y') e^{-i {\bf q}\cdot {\bf r}'} = \left(i\frac{\partial}{\partial q_x}\pm \frac{\partial}{\partial q_x}\right)e^{-i {\bf q}\cdot {\bf r}'} 
	\end{equation}
	and integrate by parts; since under our symmetry assumptions $F_H(-{\bf q},\omega)=F_H({\bf q},\omega)$ is an even function of ${\bf q}$, we get
	\begin{equation}
		\label{eq:joule_power_tmp2}
		\begin{split}
			P_{\rm Joule,\pm}(\omega) =
			4 \frac{v_0^2}{e^2\ell_B^4} &\int d^2{\bf r} \,\Re\Bigl[F_H({\bf 0},\omega) \left(x^2 + y^2\right) 
			\Bigr].
		\end{split}
	\end{equation}
	The dichroic power therefore reads
	\begin{equation}
		\label{eq:joule_dichroic_power}
		\begin{split}
			\Delta P_{\rm Joule}(\omega) =
			8 \frac{v_0^2}{e^2\ell_B^4} &\int d^2{\bf r} \,\Im\Bigl[\sigma_{xy}({\bf 0},\omega) \,{\bf r}^2
			\Bigr]
		\end{split}
	\end{equation}
	and the integrated rates difference becomes, using the Kramers-Kronig relations
	\begin{equation}
		\label{eq:joule_rates_tmp}
		\begin{split}
			\Delta \Gamma_{\rm Joule}^{\rm int} =
			4\pi \frac{v_0^2}{\hbar e^2\ell_B^4} \,\Re\Bigl[\sigma_{xy}({\bf 0},0)\Bigr] \int d^2{\bf r}\left(x^2 + y^2\right).
		\end{split}
	\end{equation}
	The small wavevector expansion of the Hall conductivity $\sigma_{xy}({\bf q},0)$ is known~\cite{HoyosSon}
	\begin{equation}
		\label{eq:hall_response_viscous_corrections}
		\begin{split}
			\sigma_{xy}(\mathbf{q},\omega=0) &= \frac{e^2}{\hbar}\frac{\mathcal C}{2\pi}\left(1 + \mathcal O({\bf q}^2) \right).
		\end{split}
	\end{equation}    
	The $\mathcal O({\bf q}^2)$ term contains information on the Hall viscosity; however, since the electric field Eq.~\eqref{eq:probe_field} is linear, all its second derivatives vanish; therefore, compatibly with Eq.~\eqref{eq:joule_rates_tmp}, no viscous correction to the current emerges. The final result reads
	\begin{equation}
		\label{eq:rate_contribution_joule}
		\begin{split}
			\Delta \Gamma^{\rm int}_{\rm Joule} =
			v_0^2\frac{16\pi}{\hbar^3} \,\int \left(
			\frac{\hbar}{8\pi\ell_B^4} \mathcal C \,{\bf r}^2 \right)d^2{\bf r}.
		\end{split}
	\end{equation}
	Notice how this ${\bf r}^2$ term, stemming from the Hall response, recovers perfectly the analogous term in the local rate expansion Eq.~\eqref{eq:viscous_marker_real_space}, upon taking the continuum limit $\sum_{\bf r} a_0^2 \simeq \int d^2{\bf r}$ to allow the complete identification of this equation with the marker decomposition of $\Delta\Gamma^{\rm int}_*$. 
	However, notice also how this is completely missing the viscous term; this latter comes from the fact that the electromagnetic perturbation Eq.~\eqref{eq:probe_field} is working against the internal stresses of the quantum Hall fluid, as we are now going to demonstrate.
	
	\subsubsection{The stress-tensor contribution}
	We begin by inspecting the strain tensor variation $\dot u_{ab}$. In the presence of the probe field Eq.~\eqref{eq:probe_field}, Eq.~\eqref{eq:strain_variation} takes the form
	\begin{equation}
		\begin{split}
			\dot u^{0}_0 =& - \dot u_{1}^{1} = \pm i\,\frac{2v_0}{e\ell_B^2B} e^{-i \omega t}
			\\
			\dot u^{0}_1 =& + \dot u_{0}^{1} = \frac{2v_0}{e\ell_B^2B} e^{-i \omega t},
		\end{split}
	\end{equation}
	which is position independent. For simplicity we single out the time-dependence $u_a^b = e^{-i \omega t}\, \tilde u_a^b$.
	
	Together with Eq.~\eqref{eq:stress_heating} and ~\eqref{eq:stress_tensor_expansion}, we get
	\begin{equation}
		\begin{split}
			P_{\rm stress,\pm}(\omega) =&
			\frac{1}{2} \int \frac{d^2{\bf q}}{4\pi^2}\int d^2{\bf r}\int d^2{\bf r}' \,e^{i {\bf q}\cdot ({\bf r}-{\bf r}')} \times\\&\times \Re\Biggl[\widetilde\eta^{aecf}({\bf q}, \omega) \, \epsilon_{eb}\epsilon_{fd} \tilde u_{c}^d(\tilde u_{a}^b)^* \Biggr]
		\end{split}
	\end{equation}
	where we introduced the Fourier transform
	\begin{equation}
		\widetilde\eta^{aecf}({\bf q}, \omega) = \int dt'\int d^2{\bf r}\,\widetilde\eta^{aecf}({\bf r}, t) e^{i \omega t} e^{-i {\bf q}\cdot {\bf r}}.
	\end{equation}
	Under our symmetry assumptions, the viscosity tensor decomposes as~\cite{Bradlyn_2012}
	\begin{equation}
		\begin{split}
			\widetilde\eta^{aecf} =& \zeta \delta^{ae}\delta^{cf} +\\+& \eta (\delta^{ac}\delta^{ef}+\delta^{af}\delta^{ec}-\delta^{ae}\delta^{cf}) +\\+&
			\overline\eta_H \frac{1}{2}\left(\epsilon^{ac}\delta^{ef}+\epsilon^{af}\delta^{ec}+a\leftrightarrow e\right)
		\end{split}
	\end{equation}
	where $\zeta$ and $\eta$ are the bulk and shear viscosity, respectively.
	Using this result, the contraction of the viscosity tensor with the strain variations simplifies considerably:
	\begin{equation}
		\begin{split}
			\widetilde\eta^{aecf}({\bf q}, \omega) \, \epsilon_{eb}\epsilon_{fd} \tilde u_{c}^d(\tilde u_{a}^b)^* =& \frac{32 v_0^2}{e^2\ell_B^4B^2} ( \eta({\bf q}, \omega)  \mp i \overline\eta_H({\bf q}, \omega)  ) \\\equiv& \frac{32 v_0^2}{\hbar^2} F_V({\bf q}, \omega).
		\end{split}
	\end{equation}
	Using this result we get
	\begin{equation}
		\begin{split}
			P_{\rm stress,\pm}(\omega) =&
			\frac{16 v_0^2}{\hbar^2} \int d^2{\bf r} \, \Re\Bigl[  F_V({\bf 0}, \omega) \Bigr].
		\end{split}
	\end{equation}
	While both the Hall and shear viscosities appear here, the dichroic power absorption is associated purely to the non-dissipative Hall one
	\begin{equation}
		\begin{split}
			\Delta P_{\rm stress}(\omega) =&
			\frac{32 v_0^2}{\hbar^3} \int d^2{\bf r} \, \Im\Bigl[  \overline\eta_H({\bf 0}, \omega)  \Bigr].
		\end{split}
	\end{equation}
	The viscous stress contribution to the transitions rates can then finally be obtained using the Kramers-Kronig relations and the fact that
	$\Re\Bigl[  \overline\eta_H({\bf 0}, 0)  \Bigr]=\overline\eta_H$:
	\begin{equation}
		\label{eq:rate_contribution_stress}
		\Delta \Gamma^{\rm int}_{\rm stress} =
		\frac{16\pi}{\hbar^3} v_0^2 \int \overline\eta_H\, d^2{\bf r}.
	\end{equation}
	This contribution, stemming from the viscous response, recovers the analogous term in the local rate expansion Eq.~\eqref{eq:viscous_marker_real_space} upon taking the continuum limit $\sum_{\bf r} a_0^2 \simeq \int d^2{\bf r}$.
	
	\subsubsection{The total rates}
	Adding up the Joule-heating Hall contribution Eq.~\eqref{eq:rate_contribution_joule}, the viscous stress one Eq.~\eqref{eq:rate_contribution_stress} we get
	\begin{equation}
		\begin{split}
			\Delta \Gamma^{\rm int}_* =& \Delta \Gamma^{\rm int}_{\rm Joule} + \Delta \Gamma^{\rm int}_{\rm stress} = \\=&
			\frac{16\pi}{\hbar^3}\,v_0^2 \int \left(
			\frac{\hbar}{8\pi\ell_B^2} \mathcal C \,{\bf r}^2 +\overline\eta_H\right)d^2{\bf r}.
		\end{split}
	\end{equation}
	This classical energy-balance calculation allows one to recover exactly the quantum-mechanical spatially resolved transition rates Eq.~\eqref{eq:viscous_marker_real_space}.
	\section{Markers and Berry curvature}
	\label{sec: markers and Berry curvature}
	The purpose of this section is to show that the local markers, such as the local Chern marker and the local Hall viscosity marker introduced in Eq.~\eqref{eq:viscosity_marker}, and the one introduced later through the approach of the rotating saddle in Eq.~\eqref{eq:viscous_marker_real_space}, arise from Berry curvatures, associated to varying a Slater determinant state.
	
	Let us assume we have a Slater determinant state
	\begin{align}
		\ket{\psi}= \bigwedge_{i=1}^{N_e} \ket{\phi_i},    
	\end{align}
	where $\ket{\phi_i}$'s are single-particle states and $N_e$ is the particle number.
	Define single-particle orthogonal projection operators by
	\begin{align}
		P = \sum_{i=1}^{N_e} |\phi_i\rangle\langle \phi_i|,
		\qquad\text{and}\qquad
		Q = 1-P.    
	\end{align}
	
	Consider Hermitian operators $\{X_\mu\}$ acting on single-particle Hilbert space. We have a family of deformed states by taking 
	\begin{align}
		\ket{\psi(\lambda)} =U(\lambda)\ket{\psi}:=\bigwedge_{i=1}^{N_e} \left(e^{i\lambda^\mu X_\mu}\ket{\phi_i}\right),
	\end{align}
	where we have introduced a deformation parameter $\lambda^{\mu}$ for each generator $X_{\mu}$, and set $\lambda=(\lambda^{\mu})$. The operator $U(\lambda)$ is the natural many-body extension of the single-particle unitary operator $e^{i\lambda^{\mu}X_{\mu}}$ and is unitary.
	
	Consider then the Berry curvature in the parameter space determined by $\lambda$, $F=dA=(1/2)F_{\mu\nu}d\lambda^{\mu}\wedge d\lambda^{\nu}$, with $A=\bra{\psi(\lambda)}d\ket{\psi(\lambda)}$. We have the following result
	\begin{proposition}
		The Berry curvature at $\lambda=0$ is
		\[
		F_{\mu\nu} =\operatorname{Tr}(P X_\mu Q X_\nu P)-\operatorname{Tr}(P X_\nu Q X_\mu P).
		\]
	\end{proposition}
	
	\begin{proof}
		We first write:
		\begin{align}
			|\psi(\lambda)\rangle = U(\lambda)|\psi\rangle.
		\end{align}
		Then the Berry connection is:
		\begin{align}
			A = \langle \psi(\lambda)| d |\psi(\lambda)\rangle = \langle \psi|U^{-1}dU|\psi\rangle.
		\end{align}
		The Berry curvature is:
		\begin{align}
			F=dA=\langle \psi| d(U^{-1}dU)|\psi\rangle.
		\end{align}
		Now using $dU^{-1}=-U^{-1}dUU^{-1}$, we can write
		\begin{align}
			d(U^{-1}dU)=dU^{-1}\wedge dU
			=-U^{-1}dU\wedge U^{-1}dU.
		\end{align}
		So:
		\begin{align}
			F&=-\langle \psi|U^{-1}dU\wedge U^{-1}dU|\psi\rangle
			\nonumber \\
			&=-\langle \psi|U^{-1}dU\,(1-|\psi\rangle\langle\psi|)\,U^{-1}dU|\psi\rangle,
		\end{align}
		where the last step follows from 
		\begin{align}
			\langle\psi|U^{-1}dU|\psi\rangle\wedge \langle\psi|U^{-1}dU|\psi\rangle =0
		\end{align}
		(due to 1-forms squaring to $0$). 
		
		At $\lambda=0$, $U=1$, and
		\begin{align}
			dU=d\bigl(1+i\lambda^\mu \widetilde X_\mu+\cdots\bigr)
			=id\lambda^\mu \,\widetilde X_\mu,
		\end{align}
		where
		\begin{align}
			\widetilde X_\mu\left(\bigwedge_{i=1}^{N_e}|\phi_i\rangle\right)
			:=\sum_{i=1}^{N_e} |\phi_1\rangle\wedge\cdots\wedge (X_\mu|\phi_i\rangle)\wedge\cdots\wedge |\phi_{N_e}\rangle.
		\end{align}
		
		So:
		\[
		F=\langle\psi|\widetilde X_\mu(1-|\psi\rangle\langle\psi|)\widetilde X_\nu|\psi\rangle \, d\lambda^\mu\wedge d\lambda^\nu.
		\]
		
		Now:
		\begin{align}
			\langle\psi|\widetilde X_\mu|\psi\rangle = \sum_{i=1}^{N_e} \langle\phi_i|X_\mu|\phi_i\rangle
		\end{align}
		and
		\begin{align}
			\langle\psi|\widetilde X_\mu\widetilde X_\nu|\psi\rangle
			&=\sum_{i=1}^{N_e}\langle\phi_i|X_\mu X_\nu|\phi_i\rangle
			\nonumber \\
			&-\sum_{i\neq j}\langle\phi_i|X_\mu|\phi_j\rangle\langle\phi_i|X_\nu|\phi_i\rangle \nonumber \\
			&+\sum_{i\neq j}\langle\phi_i|X_\mu|\phi_i\rangle\langle\phi_j|X_\nu|\phi_j\rangle.
		\end{align}
		Also,
		\begin{align}
			\langle\psi|\widetilde X_\mu|\psi\rangle\langle\psi|\widetilde X_\nu|\psi\rangle
			&=\sum_{i\neq j}\langle\phi_i|X_\mu|\phi_i\rangle\langle\phi_j|X_\nu|\phi_j\rangle
			\nonumber \\
			&+\sum_{i=1}^{N_e}\langle\phi_i|X_\mu|\phi_i\rangle\langle\phi_i|X_\nu|\phi_i\rangle.
		\end{align}
		
		So:
		\begin{align}
			&\langle\psi|\widetilde X_\mu\widetilde X_\nu|\psi\rangle
			-\langle\psi|\widetilde X_\mu|\psi\rangle\langle\psi|\widetilde X_\nu|\psi\rangle \nonumber \\
			&=\sum_{i=1}^{N_e}\langle\phi_i|X_\mu X_\nu|\phi_i\rangle
			-\sum_{i=1}^{N_e}\sum_{j=1}^{N_e}\langle\phi_i|X_\mu|\phi_j\rangle\langle\phi_j|X_\nu|\phi_i\rangle \nonumber \\
			&=\sum_{i=1}^{N_e}\langle\phi_i|X_\mu\left(1-\sum_{j=1}^{N_e}|\phi_j\rangle\langle\phi_j|\right)X_\nu|\phi_i\rangle
			\nonumber \\
			&=\operatorname{Tr}(P X_\mu Q X_\nu P).
		\end{align}
		
		Thus
		\begin{align}
			F&=\operatorname{Tr}(P X_\mu Q X_\nu P)\,d\lambda^\mu\wedge d\lambda^\nu
			\nonumber \\
			&=\frac{1}{2}\Bigl(\operatorname{Tr}(P X_\mu Q X_\nu P)-\operatorname{Tr}(P X_\nu Q X_\mu P)\Bigr)
			\,d\lambda^\mu\wedge d\lambda^\nu.
		\end{align}
		as claimed.
	\end{proof} 
	Now, if the system is a lattice, with single-particle Hilbert space generated by states $\ket{\bm{\rho}}$, with $\bm{\rho}$ a vector labeling points on the lattice, we can write
	\begin{align}
		F_{\mu\nu} &=\tr(P X_\mu Q X_\nu P)-\tr(P X_\nu Q X_\mu P)\nonumber \\
		&=2i \; \operatorname{Im}\tr(PX_{\mu}QX_{\nu}P)\nonumber \\
		&= 2i \sum_{\bm{\rho}} \operatorname{Im}\bra{\bm{\rho}}PX_{\mu}QX_{\nu}P\ket{\bm{\rho}}.
	\end{align}
	Now the function under the summation sign has the form of the local markers (up to a multiplicative constant independent of $\bm{\rho}$), for different choices of generators of deformations. For example if $X_{\mu}$ are the position operators, we recover, the local Chern markers. If  the $X_{\mu}$ are taken to be $G_0,G_1$, we recover the local viscosity marker, and if they are taken to be $\mathcal{Q}_0$ and $\mathcal{Q}_1$, we recover the local viscosity marker corresponding to rotating saddle approach.
	
	\section{Frequency selection}
	We give here a more detailed explanation of the frequency resolved measurement briefly presented in the main text.
	
	In particular, we focus on the low-energy, small magnetic field $\ell_B\gg a_0$ limit, 
	where the lower bands of the Harper-Hofstadter Hamiltonian Eq.~\eqref{eq:hh_hamiltonian} become roughly flat, with the $n-$th band sitting at energy $E_n\simeq -2(J_x+J_y)+\hbar\omega_c\left(n+\frac12\right)+\mathcal{O}(\alpha^2)$
	(cf. Eq.~\eqref{eq:hh_limit_sm}).
	In this same limit, the analysis carried out in Sec.~\ref{smsec:marker_analysis} holds. 
	In particular, 
	the position operators appearing in the real space rotating-saddle probe Eq.~\eqref{eq:rotating_saddle} (and thus also in the associated marker Eq.~\eqref{eq:viscous_marker_real_space}) can be split according to Eq.~\eqref{eq:saddle_decomposition}:
	it can be seen from the table~\ref{tab:operators_list} that the transitions associated to the Hall viscosity (induced by $\mathcal D_1$ and $\mathcal E_1$) satisfy $|\Delta n|=2$ and thus have energy $\Delta E_V\simeq2\hbar\omega_c$; on the other hand, the transitions associated to the Hall contribution (induced by $\mathcal D_2$ and $\mathcal E_2$) satisfy $|\Delta n|=|\Delta m|=1$ and therefore have energy $\Delta E_H\simeq\hbar\omega_c = \frac12\Delta E_V$.
	This difference is what allows for the frequency separation of the two contributions.
	
	We conclude the section by analyzing how (weak) lattice effects modify this picture. 
	
	We begin by noticing how the matrix elements $\braket{\beta|\mathcal Q_0\mp i \mathcal Q_1|\alpha}$ appearing in the Fermi's golden rule expression for the transition rates for the non-interacting system
	\begin{equation}
		\begin{split}
			\Gamma_\pm(\omega) &= 2\pi \frac{|v_0|^2}{4\hbar^2}\sum_{\alpha\in {\rm occ.}}\sum_{\beta\in{\rm emp.}} |\braket{\beta|\mathcal Q_0\mp i \mathcal Q_1|\alpha}|^2\delta(\omega_{\beta,\alpha}-\omega)
		\end{split}
	\end{equation}
	can be written as
	\begin{equation}
		\label{eq:matrix_elements_decomposition}
		\braket{\beta|\mathcal Q_0\mp i \mathcal Q_1|\alpha} =
		\ell_B^2\sum_{j=1}^3\braket{\beta|\mathcal D_j\mp i \mathcal E_j|\alpha}.
	\end{equation}
	In the following, we will neglect the edge contributions arising from the $\mathcal D_3$ and $\mathcal E_3$ operators for simplicity.
	
	We notice how Eq.~\eqref{eq:matrix_elements_decomposition} should be computed using the eigenstates $\ket{\alpha}$ of Eq.~\eqref{eq:hh_magnetic_translations}, which includes the effects of the underlying lattice;
	in this regard we notice how the Hamiltonian Eq.~\eqref{eq:hh_magnetic_translations} has a conserved quantity:
	we can indeed introduce the following parity operator
	\begin{equation}
		\mathcal P = e^{i \pi a^\dagger a};
	\end{equation}
	notice how this operator, in the low-flux limit $\ell_B\gg a_0$, simply counts the band index $n$.
	Notice also how this operator satisfies
	\begin{equation}
		\mathcal P^\dagger \Pi_{x(y)} \mathcal P = - \Pi_{x(y)};
	\end{equation}
	as a consequence, $\mathcal P^\dagger T_{(a_x,a_y)} \mathcal P=T^\dagger_{(a_x,a_y)}$ and thus
	\begin{equation}
		[H,\mathcal P] = 0.
	\end{equation}
	It follows that the eigenstates of $H$ can be chosen to have a well-definite parity. 
	
	It follows that the matrix elements
	\begin{equation}
		\braket{\beta|\mathcal D_1\mp i \mathcal E_1|\alpha},
	\end{equation}
	which as discussed in Sec.~\ref{smsec:marker_analysis} contain the viscous contribution to Eq.~\eqref{eq:viscous_marker_real_space}, 
	are potentially non-vanishing only if $\ket{\alpha}$ and $\ket{\beta}$ have the same parity: 
	as a consequence, the associated Bohr frequencies $\omega_{\beta,\alpha}$ are (roughly) even multiples of $\hbar\omega_c$.
	
	On the other hand the matrix elements
	\begin{equation}
		\braket{\beta|\mathcal D_2\mp i \mathcal E_2|\alpha},
	\end{equation}
	which encode the Hall contribution instead, 
	are potentially non-zero only if $\ket{\alpha}$ and $\ket{\beta}$ have opposite parity: 
	as a consequence, the associated Bohr frequencies $\omega_{\beta,\alpha}$ are (roughly) odd multiples of $\hbar\omega_c$.
	
\end{supplementalMaterials}

\end{document}